\documentclass[prd,twocolumn,superscriptaddress,floatfix,amsmath,amssymb]{revtex4-1}

\usepackage[utf8]{inputenc}
\usepackage{amsmath}
\usepackage{amssymb}
\usepackage{amsthm}
\usepackage{braket}
\usepackage{ifthen}
\usepackage{xcolor}
\usepackage{mathtools}
\usepackage{graphicx}
\usepackage{natbib}
\usepackage{bm}
\usepackage{enumerate}
\usepackage{hyperref}
\usepackage{savesym}
\usepackage{makecell}
\usepackage{scalerel}
\usepackage{centernot}

\newcommand{\ketbra}[2]{\ket{#1}\!\bra{#2}}
\DeclareMathOperator{\Tr}{Tr}
\DeclareMathOperator{\dom}{dom}
\DeclareMathOperator{\Span}{Span}
\DeclareMathOperator{\om}{O}
\DeclareMathOperator{\oa}{O_2}
\DeclareMathOperator{\ob}{O_1}

\newtheorem{definition}{Definition}[section]  
\newtheorem{theorem}{Theorem}[section]  
\newtheorem{lemma}{Lemma}[section]

\newcommand{\po}{\hat\Pi}
\newcommand{\id}{\hat\openone}
\newcommand{\Pa}{\hat P^{(2)}}

\newcommand{\poa}{\po^{(2)}}
\newcommand{\pob}{\po^{(1)}}

\newcommand{\m}{\mathcal{C}}
\newcommand{\ma}{\mathcal{C}_2}
\newcommand{\mb}{\mathcal{C}_1}
\newcommand{\pa}{p^{(2)}}
\newcommand{\pb}{p^{(1)}}
\newcommand{\Va}{V^{(2)}}
\newcommand{\Vb}{V^{(1)}}
\newcommand{\pVa}{\pa_j,\Va_j}
\newcommand{\pVb}{\pb_i,\Vb_i}

\newcommand{\coarser}{\hookrightarrow}
\newcommand{\finer}{\hookleftarrow}

\newcommand{\pcoarser}[1]{\lhook\joinrel\xrightarrow{#1}}
\newcommand{\pfiner}[1]{\xleftarrow{#1}\joinrel\rhook}

\begin{document}

\title{Data processing makes POVMs coarser and observational entropies larger}
\author{Adam Teixid\'o-Bonfill}
 \affiliation{Department of Applied Mathematics, University of Waterloo, Waterloo, Ontario, N2L 3G1, Canada}
\affiliation{Institute for Quantum Computing, University of Waterloo, Waterloo, Ontario, N2L 3G1, Canada}
\affiliation{Perimeter Institute for Theoretical Physics, 31 Caroline St N, Waterloo, Ontario, N2L 2Y5, Canada}

\begin{abstract}
We find a criterion to compare POVM measurements and decide which ones can extract more information from physical systems, with coarser POVMs always extracting less information. This criteria generalizes the previous definition of coarser POVM, and is motivated by the idea that information cannot be gained by processing the measurement outcomes. The information that a measurement cannot extract is quantified by observational entropy or coarse-grained entropy. Adequately, coarser measurements have larger observational entropies. Moreover, the characterization and properties of coarser measurements that we provide allow to straightforwardly derive several previously known results about observational entropy.
\end{abstract}

\maketitle

\section{\textbf{Introduction}} 

Our limited capacity to measure physical systems makes all measurements always have some uncertainty. This uncertainty has multiple sources, including experimental error and limited resolution. On top of this, quantum mechanics has unavoidable uncertainties, conjugate quantities such as position and momentum cannot be simultaneously known with perfect accuracy. Quantifying these uncertainties becomes crucial, because uncertainties constrain our predictions about physical systems.

Positive operator-valued measure (POVM) measurements are the best framework to describe the uncertainties that can occur in quantum measurements. POVM measurements are the most general type of measurements and extend the common projective measurements~\cite{nielsen_chuang_2010}. Precisely, projective measurements assume that measurement apparatus are completely reliable. On the contrary, POVM measurements incorporate in their formalism uncertainties and errors. Moreover, POVM measurements naturally arise when measuring systems through probes that interact with them. In general, measuring a quantum system modifies it, but the POVMs do not specify how. The POVMs together with the updates of the states are called generalized measurements, or quantum instruments as well \cite{Davies1970,Pellonp__2012}. 

Coarse-grainings capture an ubiquitous source of uncertainty, having limited resolution. The concept of coarse-graining appears in multiple branches of physics under different guises. These branches go from consistent histories quantum mechanics \cite{Consistent1993,Consistent1996,Consistent2019}, to fluid dynamics \cite{fluids1997,fluids2017}, black holes \cite{black2019,black2021} and the renormalization group \cite{renormalization1998,renormalization2007}. Moreover, several entropies incorporate a notion of coarse-graining, including the entropy of an observable \cite{observable1962,observable1984,observable1997,observable2017,observable2019}, Kolmogorov-Sinai entropy \cite{doyne1982,PhysRevLett.82.520,kolmogorov2004,jost2005dynamical} and topological entropy \cite{doyne1982,jost2005dynamical}. Coarse-grainings consist of only having access to limited features of a system, e.g. whether an electron is on the right or the left of a box. This limitation often makes several states indistinguishable. Particularly, measurements can be coarse-grained, as done to define coarse-grained entropy \cite{vonNeumann,wehrl1978,black2021} or observational entropy \cite{safranek2019thermRapid,safranek2019therm,brief2021}. This entropy quantifies the uncertainty perceived by an observer that can only perform a limited set of measurements. Observational entropy recently arose as a promising candidate for thermodynamic entropy \cite{PRXQuantum.2.030202,safranek2019thermRapid,safranek2019therm,strasberg2019,classical2020}, while the origin of this entropy dates back to as early as the first half of the previous century, when introduced by Wigner and von Neumann \cite{vonNeumann2010}.

The amount of coarse-graining varies among measurements, with lower resolution measurements being more coarse-grained, or coarser. Coarser measurements should distinguish less among states and correspondingly have a higher observational entropy. Being coarser has been defined for projective measurements \cite{safranek2019therm,PRXQuantum.2.030202}, connecting them to the idea of coarser partitions. Afterwards, the notion of coarseness was brought to POVM measurements \cite{safranek2021information}. 

In this work, we provide a notion of coarser POVM measurements that generalizes the previous notion. This generalization comes from realizing the following: processing the outcomes of measurements never increases the information that measurements gather. Therefore, the processed data can be considered the outcomes of a new, coarser measurement. Then, this connection to processing allows to promptly derive the properties of coarser measurements using well-known data processing inequalities. Crucially, observational entropy does become larger for coarser measurements. Moreover, all measures of information that satisfy data processing inequalities become monotonic, including relative entropy, Rényi entropies and mutual information. Therefore, the coarseness of measurements indicates their uncertainty independently of how the unknown information is quantified.

This article aims to present the generalized definition of coarser POVM measurements and the major properties of the definition, including its connection to data processing. Section \ref{sec:background} formally describes generalized measurements, POVMs, observational entropy, and data processing inequalities. Then, section \ref{sec:CoarserAndProcessing} defines coarser POVM, connecting the definition to data processing and proving properties of the definition. Moreover, the concept of coarseness is further extended to describe POVMs that are coarser only in subspaces of states. Afterwards, section \ref{sec:obsEntropyProp} uses coarser measurements to quickly prove known important properties of observational entropy. The article concludes connecting the main results to other areas of physics.

\section{\textbf{Background concepts}}
\label{sec:background}

This section starts describing generalized measurements and then summarizes how to extend observational entropy to generalized measurements as proposed by Šafránek and Thingna \cite{safranek2021information}. Afterwards, the text connects observational entropy to widespread informational entropies, providing their interpretations. Finally, this section finishes modelling data processing as Markov chains. This modelling allows to show that processing increases observational entropy and decreases mutual information. These results translate to properties of coarser measurements in section \ref{sec:CoarserAndProcessing}.

\textit{\textbf{Generalized measurements.-}}
Also called quantum instruments, these measurements are the most general quantum measurements. Generalized measurements are described by a set of quantum operations $\{\mathcal{A}_i\}$ such that $\sum_i \Tr[\mathcal{A}_i(\hat\rho)] = 1$ for all system states $\hat\rho$. These $\mathcal{A}_i$ can always be decomposed as
\begin{equation}
    \mathcal{A}_i(\hat\rho) = \sum_m \hat K_{im} \hat\rho \hat K_{im}^\dagger,
\end{equation}
where $\{\hat K_{im}\}$ are linear operators called Kraus operators. These operators fulfill
\begin{equation}
    \sum_{i, m } \hat K_{im}^\dagger \hat K_{im} = \id,
\end{equation}
which is equivalent to $\sum_i \Tr[\mathcal{A}_i(\hat \rho)] = 1, \  \forall\hat\rho$. The values of $i$ are the outcomes, which occur with probability
\begin{equation}
    p_i = \Tr[\mathcal{A}_i(\hat \rho)] = \Tr[\po_i\hat\rho].
\end{equation}
Here, $\po_i = \sum_m\hat K_{im}^\dagger \hat K_{im}$ are positive operators, which form a positive operator-valued measure (POVM). After obtaining $i$, the state is updated according to
\begin{equation}
    \hat\rho \overset{i}{\longrightarrow} \frac{\mathcal{A}_i(\hat\rho)}{p_i}.
\end{equation}
This update affects the results of any measurements performed afterwards. Importantly, the $\po_i$ do not specify the post-measurement state, because the decomposition of the $\po_i$ in terms of Krauss operators is not unique. Nonetheless, the $\po_i$ are enough to define coarser measurements, because the $\po_i$ already specify the measurement statistics. Kraus operators are only used in section \ref{sec:obsEntropyProp}, to recover properties of observational entropy from properties of coarser measurements. Finally, projective measurements are a particular case of generalized measurements. Measurements are projective when the $\po_i$ are projectors $\hat P_i$, and each projector has a single Kraus operator $\hat K_i = \hat P_i$.

\textit{\textbf{Observational entropy and its connection to relative entropy.-}} The observational entropy of a state $\hat\rho$ depends on which generalized measurement $\m$ is available or relevant. Namely,
\begin{equation}
    S_{\m}(\hat\rho) = \sum_i p_i [\ln{V_i} - \ln{p_i}],\label{eq:SCrho}
\end{equation}
where $p_i=\Tr[\hat\rho \po_i]$ and $V_i = \Tr \po_i$. This definition was proposed by Šafránek and Thingna~\cite{safranek2021information}, and extends the earlier definition of observational entropy for multiple projective measurements \cite{safranek2019thermRapid,safranek2019therm}, while retaining the important properties reviewed in section~\ref{sec:obsEntropyProp}.

The $V_i$ are the volumes of the subspaces of the outcomes $i$ and quantify how often $i$ is measured on random states. 
Concretely,
\begin{equation}
    V_i = \Tr[\po_i \hat\rho^{\text{id}}]V_\mathcal{H} = p_i^{\text{id}}V_\mathcal{H}, \label{eq:volId}
\end{equation}
where $p_i^{\text{id}}$ is the probability of measuring $i$ when the state is uniformly random, $\hat\rho^{\text{id}}=\id/V_\mathcal{H}$, and $V_\mathcal{H}=\dim{\mathcal{H}}$ is the volume of the Hilbert space. This relation justifies that $V_i$ behave as probabilities, which is crucial to show that $S_\m$ is larger for coarser measurements in the next section.

The following paragraphs connect $S_{\m}$ to well known information theory quantities. This connection will prove useful afterwards. The first quantity is the von Neumann entropy \mbox{$S_{\text{vN}}(\hat\rho) = -\Tr[\hat\rho\ln\hat\rho]$}, which measures the unknown information in the system for an observer who knows $\hat\rho$. When $\m$ is a projective measurement with projectors $\hat P_i$, then
\begin{equation}
    S_{\m}(\hat\rho) = S_{\text{vN}}\left(\sum_{i} p_i \frac{\hat P_i}{V_i}\right). \label{eq:interpret}
\end{equation}
This identity and similar ones were key to derive the second law of thermodynamics with observational entropy \cite{strasberg2019,PRXQuantum.2.030202}. Moreover, this equation allows to interpret observational entropy when $\m$ is projective: $S_{\m}(\hat\rho)$ quantifies the unknown information from the perspective of an observer who knows the $p_i$ and $\hat P_i$ but not the state $\hat\rho$. This observer could obtain the $p_i$ by performing their only available measurement $\m$ on infinite copies of $\hat\rho$. Then, of all the states with the same $p_i$, $\hat\rho_{\text{est}}=\sum_{i} p_i \hat P_i/V_i$ has the largest entropy. Therefore, $\hat\rho_{\text{est}}$ is the observer's best estimation of $\hat\rho$, by the principle of maximum entropy. Finally, $S_{\text{vN}}(\hat\rho_{\text{est}})$ is the unknown information in $\hat\rho_{\text{est}}$, and $S_{\text{vN}}(\hat\rho_{\text{est}})=S_{\m}(\hat\rho)$ justifies the given interpretation of observational entropy. However, neither this interpretation nor Eq.~\eqref{eq:interpret} extend to generalized measurements. This claim is shown in appendix \ref{apx:failedvNrel}, prompting that the interpretation of $S_\m$ has to be modified for general $\m$.

For generalized measurements and finite dimensional Hilbert spaces, an alternative to Eq.~\eqref{eq:interpret} is
\begin{align}
    S_{\m}(\hat\rho) &= \sum_i p_i [\ln(p_i^{\text{id}}V_\mathcal{H}) - \ln{p_i}]\nonumber\\
     &= \ln V_\mathcal{H} - \sum_i p_i [\ln{p_i}-\ln p_i^{\text{id}}]\nonumber\\
     &= S_{\text{vN}}(\hat\rho^{\text{id}}) - D_{\text{KL}}\left[p_i || p_i^{\text{id}}\right], \label{eq:interpretacioinfo}
\end{align}
which was already provided in \cite{safranek2019thermRapid,safranek2019therm}.
Here, $D_{\text{KL}}$ denotes the Kullback-Leibler divergence, also called relative entropy, 
\begin{equation}
    D_{\text{KL}}\left[p_i || q_i\right] = \sum_i p_i (\ln{p_i} - \ln{q_i}).
\end{equation}
This divergence quantifies the extra information needed to transmit $i$ sampled with probabilities $p_i$, due to communicating with a code optimized for probabilities $q_i$ instead. These $q_i$ usually represent the best available model for the probabilities. Therefore, $D_{\text{KL}}$ measures how different the actual $p_i$ are from the reference $q_i$, as supported by the Gibbs' inequality, 
\begin{equation}
    D_{\text{KL}}\left[p_i || q_i\right]\geq 0, \quad D_{\text{KL}}\left[p_i || q_i\right]=0 \iff p_i=q_i,\  \forall i. \label{eq:gibbs}
\end{equation}
The meanings of $S_{\text{vN}}$ and $D_{\text{KL}}$ provide a modified interpretation of $S_\m$ valid for general $\m$: $S_\m(\hat\rho)$ quantifies the unknown information for an observer who only knows $p_i$ and $V_i$ and only considers $\m$ relevant. This observer ignores everything else about the $\po_i$ and $\hat\rho$. Therefore, this observer cannot restrict in any way the density matrix of the state and must model it as $\hat\rho^{\text{id}}$. However, the observer does have information about the state, since the observer knows that the outcomes of $\m$ will be distributed with probabilities $p_i$ instead of $p_i^{\text{id}}=V_i/V_{\text{tot}}$. Then, looking at the terms in Eq.~\eqref{eq:interpretacioinfo},
\begin{itemize}
    \item $S_{\text{vN}}(\hat\rho^{\text{id}})$ is the information that the observer ignores about the density matrix of the system. 
    \item $D_{\text{KL}}\left[p_i || p_i^{\text{id}}\right]$ is the information that the observer can avoid sending when they transmit the outcomes of $\m$, using that they know the probabilities $p_i$.
\end{itemize}
The net unknown information for the observer, for whom only $\m$ is relevant, is precisely \mbox{$S_{\text{vN}}(\hat\rho^{\text{id}})-D_{\text{KL}}\left[p_i || p_i^{\text{id}}\right]=S_\m(\hat\rho)$}, justifying the previously stated interpretation.


\textit{\textbf{Data processing model and inequalities.-}} The last task in this section is to model data processing as a Markov chain. Consider two random variables $I$ and $J$; for instance, the output of two different measurements. If these variables form a Markov chain, denoted by $I \to J$, then the conditional probabilities $p_{j|i}$ are fixed. These lowercase $i$ and $j$ respectively indicate any value that $I$ and $J$ can take. The conditional probabilities form a transition matrix $P_{ji} = p_{j|i}$, which is left stochastic, meaning
\begin{equation}
    \sum_j P_{ji} = 1,\ P_{ji}\geq 0.
\end{equation}
Importantly, the probabilities $p_i$ of $I$ are variable and choosing them fixes the probabilities $p_j$ of $J$. This relationship spontaneously appears if the $j$ are produced by a black box that only has access to $i$ and a source of randomness. This setup is visualized in figure \ref{fig:DataProc}, illustrating why $I\to J$ can be understood as describing data processing.
\begin{figure}[ht]
    \centering
    \includegraphics[width=6cm]{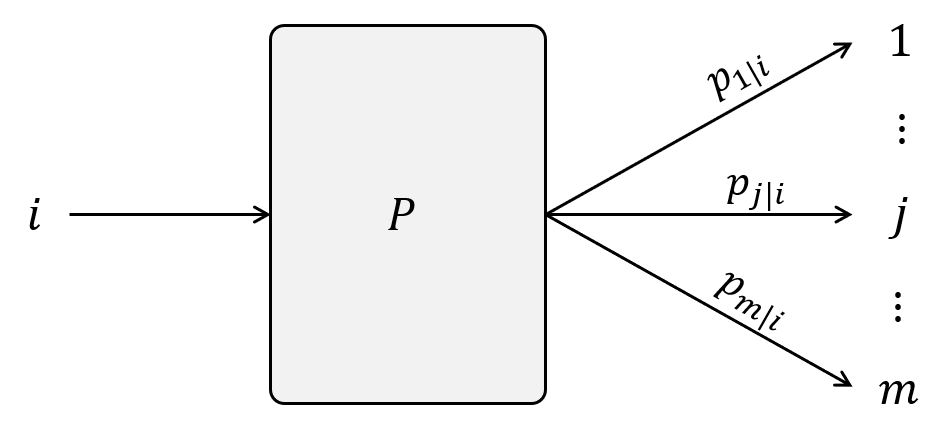}
    \caption{Black box example of data processing. The transition matrix $P_{ji} = p_{j|i}$ is enough to specify the probabilities of the outputs $j \in \{1,\ldots,m\}$, for any given input $i$. }
    \label{fig:DataProc}
\end{figure}

Moreover, this model can be further extended to include the measured physical system. The probabilistic state of the system is then represented by a third random variable $X$. The value taken by $X$ must determine the outcome of any measurement $Y$, therefore $X\to Y$. Afterwards, if $Y$ is processed to obtain $Z$ without measuring the system again, then $p_{xyz} = p_{z|y}p_{xy}$. Altogether, this setup forms a Markov chain of three variables, $X \to Y \to Z$. For example, in quantum systems, the values of $X$ can be the the eigenvectors of $\hat \rho = \sum_x p_x \ketbra{x}{x}$. Then, $p_{y|x} = \bra{x}\po^{(Y)}_y\ket{x}$ and the $p_{z|y}$ characterize how $Y$ is processed to obtain $Z$, in the same manner already shown in figure \ref{fig:DataProc}.

Next, several inequalities about processing are proven relying on the relation between processing and Markov chains. Afterwards, these inequalities will be used to show fundamental properties of coarser measurements.
\begin{theorem} {\normalfont (Data processing inequality for the $D_{\text{KL}}$)} Consider a Markov chain $I\to J$ and two probability distributions $p_i$ and $q_i$, then
\begin{equation}
    D_{\text{KL}}[p_i||q_i] \geq D_{\text{KL}}[p_j||q_j]. \label{eq:dklMonotonous}
\end{equation}
The equality is achieved if and only if $p_{ij} = q_{i|j}p_j$.
\label{thm:dataProcStr}
\end{theorem}
This inequality, proven in appendix \ref{apx:dpistr}, implies that processing cannot make probability distributions more distinguishable. This result is well known, but takes different names, even having been referred to as a second law of thermodynamics \cite{elements}. The justification for this name is that relative entropy never increases along Markov chains. Moreover, these Markov chains could represent physical processes progressing. Meanwhile, entropy surprisingly can both increase and decrease along a Markov chain. Interestingly, Eq.~\eqref{eq:dklMonotonous} holds for the general class of Rényi relative entropies \cite{Renyi}. In addition, the inequality can be modified to apply for \mbox{f-divergences} \cite{fdivergences} or to obtain a quantum version, which is known as monotonicity of the relative entropy \cite{Lindblad1975,Uhlmann1977}. These extensions are not explored in this article for conciseness. However, these extensions could be translated to properties of coarser measurements, similarly to what is done in section \ref{sec:CoarserAndProcessing} for observational entropy. 

The next inequality states that observational entropy cannot decrease with processing, even using the definition
\begin{equation}
    S_{\text{obs}}(p_i,V_i) = \sum_i p_i [\ln V_i - \ln p_i].
    \label{eq:obspV}
\end{equation}
This definition does not only apply to quantum measurements, but also to classical measurements \cite{wehrl1978,PhysRevLett.82.520,doi:10.1119/1.1632488,doi:10.1063/1.2907731,upanovi2018}, where the definition has been considerably explored \cite{Kozlov2007,classical2020}.
\begin{theorem} {\normalfont (Observational entropy does not decrease along Markov chains)} Consider
\begin{align}
    p_j &= \sum_i P_{ji} p_i,\
    V_j = \sum_i P_{ji} V_i,
    \label{eq:probrel}
\end{align}
with $P$ left stochastic. Then,
\begin{equation}
    S_{\text{obs}}(p_i,V_i) \leq S_{\text{obs}}(p_j,V_j), \label{eq:sobsMonotonous}
\end{equation}
with equality if and only if  $P_{ji}p_{i}V_j=P_{ji}V_ip_j$. 

Moreover, Eq.~\eqref{eq:probrel} is equivalent to having $I\to J$ with transition matrix $P$ and the $V_i$ transforming as the $p_i$.
\label{thm:observationalmonotonous}
\end{theorem}
The requirement that volumes transform as probabilities is sensible, because $V_i/V_{\text{tot}}=p_i^{\text{id}}$, where \mbox{$V_\text{tot} = \sum_i V_i$} and $p_i^{\text{id}}$ are the probabilities for the uniformly random state, analogously to Eq.~\eqref{eq:volId}. The theorem is proven in appendix \ref{apx:dpiobs} and shows that when $p_i$ and $p^{\text{id}}_i$ evolve under the same Markov chain, observational entropy never decreases. Note that the converse does not hold, one might encounter a pair $p_i$ and $p^{\text{id}}_i$ for which observational entropy does not decrease, but the process that occurs is not a Markov chain. The adequate counterexamples are provided in appendix \ref{apx:counterexampleConverse}.

The last important inequality is often referred to as the data processing inequality \cite{elements}. This inequality concerns the mutual information of a joint probability distribution $p_{xy}$,
\begin{equation}
    I(p_{xy}) = \sum_{x, y} p_{xy} \ln{\frac{p_{xy}}{p_x p_y}},
\end{equation}
where $p_x$ and $p_y$ are the marginals of $p_{xy}$.

\begin{theorem} {\normalfont(Data processing inequality)} Consider the Markov chain $X\to Y\to Z$. Then,
\begin{equation}
    I(p_{xy}) \geq I(p_{xz}).
\end{equation}
The equality is achieved if and only if  $p_{xyz} = p_{y|z}p_{xz}$.
\label{thm:DPIMI}
\end{theorem}
This theorem, proven in appendix \ref{apx:dpimi},  entails that processing the outcomes of measurements cannot increase the information that these measurements already extracted. To see this implication, consider modelling the state of the system as $X$, the measurement outcomes as $Y$ and the processed measurement outcomes as $Z$. This situation corresponds to $X\to Y\to Z$, as we previously discussed when providing a data processing model. Then, the information that $Z$ conveys about $X$ cannot be larger than the information that $Y$ conveys about $X$, due to the data processing inequality.

The next section uses the two inequalities that concern observational entropy and mutual information, through connecting coarser to data processing.

\section{Generalizing the definition of coarseness in quantum mechanics}
\label{sec:CoarserAndProcessing}

The notion of coarser measurements in quantum mechanics comes from projective measurements. For these measurements, the notion of coarser can be traced back to the common idea of coarser partitions. Specifically, the subspaces associated to the projectors form partitions of the Hilbert space. Then, projective measurements are coarser if their corresponding partitions are coarser. However, most generalized measurements are not in correspondence with partitions, providing an opportunity to extend the definition of coarser. Šafránek and Thingna~\cite{safranek2021information} previously provided a possible extension: $\ma$ is coarser than $\mb$ if their $\{\poa_j\}$ and $\{\pob_i\}$ fulfill
\begin{equation}
    \poa_j = \sum_{i\in I_{j}} \pob_i,\label{eq:coarserold}
\end{equation}
where $I_{j}$ are disjoint sets of outcomes and $\cup_j I_{j} = \{i\}$, with $ \{i\}$ indicating the set of all possible $\mb$ outcomes. Note that superindices will be used throughout the text to specify the measurement. 
The above equation implies that the $\{\poa_j\}$ can be built in terms of the $\{\pob_i\}$. Nonetheless, there are more ways in which the $\{\poa_j\}$ could be built. These ways are included in the generalization below, which is the central idea of the present work.
\begin{definition} {\normalfont (Coarser measurement)} Consider two generalized measurements $\ma$ and $\mb$. Then, $\ma$ is coarser than $\mb$, denoted by $\ma \coarser \mb$, when
\begin{align}
     \poa_j &= \sum_i P_{ji}\ \pob_i,\label{eq:coarser}
\end{align}
where $P$ must be a left stochastic matrix, i.e. $\sum_j P_{ji} = 1$, $P_{ji}\geq0$.
\label{def:coarser}
\end{definition}
Allowing for mixtures of the POVM elements in Eq.~\eqref{eq:coarser} better matches that POVM elements can be expressed as mixtures of projectors. The following properties are expected when $\ma \coarser \mb$,
\begin{itemize}
    \item Observational entropy is larger for the coarser measurement, \mbox{$S_{\ma}(\hat\rho)\geq S_{\mb}(\hat\rho)$} $\forall\hat\rho$.
    \item Measuring the coarser $\ma$ gives less information about $\hat\rho$ than measuring $\mb$.
\end{itemize} 
Indeed, these properties are respectively shown in theorems \ref{thm:CoarserObservational} and \ref{thm:CoarserProcessing}.

The new definition of coarser successfully generalizes the previous one in Eq.~\eqref{eq:coarserold}. The previous definition is recovered by restricting $P$ to fulfill
\begin{equation}
    P_{ji} = 
    \begin{cases}
    1 & \text{for}\ i\in I_{j},\\
    0 & \text{otherwise} ,
    \end{cases}
\end{equation}
where the $I_{j}$ are disjoint and $\cup_j I_{j} = \{i\}$. Furthermore, both definitions of coarser are equivalent when the coarser measurement $\ma$ is projective, as stated below and proven in appendix \ref{apx:generalized}.
\begin{theorem}
For any $\ma$ with projectors $\{\Pa_i\}$ and any $\mb$ with POVM elements $\{\pob_j\}$, 
\begin{equation}
    \ma \coarser \mb \iff \Pa_j = \sum_{i\in I_{j}} \pob_i 
\end{equation}
where $I_{j}$ are disjoint.
\label{thm:projective}
\end{theorem}
Otherwise, when $\ma$ is not projective, the new definition is richer. For instance, for any projective measurement $\{\hat P_i\}$, define the measurement $\{\po_j\}$ as
\begin{equation}
    \po_j = \sum_i P_{ji} \hat P_i,
\end{equation}
with $P$ left stochastic. Then, $\{\po_j\}\coarser\{\hat P_i\}$, but with the earlier definition $\{\po_j\}\not\coarser\{\hat P_i\}$ for almost all $P$, specifically for the ones that have entries other than $0$ or $1$. Altogether, the definition of coarser proposed in this work generalizes the earlier definition, allowing to compare considerably more measurements.

\textit{\textbf{Properties of coarser measurements and relation to data processing.-}}
Coarser measurements can be described in terms of data processing, as shown next. Adopting this point of view allows to extend the definition of coarser to any measurement, even outside quantum theory. However, the following results focus on generalized measurements for the sake of simplicity. 

\begin{lemma} {\normalfont(Processing makes measurements coarser)} Consider two generalized measurements $\mb$, $\ma$ and denote their outcome probabilities by $\pb_i$, $\pa_j$. Then, the following statements are equivalent,
\begin{enumerate}[(i)]
    \item $\ma\coarser\mb$.\label{item:coarser}
    \item There is a left stochastic $P$ such that \begin{equation}
        \pa_j = \sum_i P_{ji}\ \pb_i,\ \forall\hat\rho.
    \end{equation} \label{item:probrel}
    \item There is a Markov chain $I\to J$ such that if $I$ is distributed as $\pb_i$ then $J$ is distributed as $\pa_j$, $\forall\hat\rho$.\label{item:ItoJ}
\end{enumerate} 
\label{lemma:processingmakescoarser}
\end{lemma}
\begin{proof}
Starting with $(\ref{item:coarser})$, $\ma\coarser\mb$,  Eq.~\eqref{eq:coarser} is equivalent to
\begin{align}
     \Tr[\poa_j\hat\rho] &= \Tr\Big[\sum_i P_{ji}\pob_i\hat\rho\Big] ,\  \forall\hat\rho,
\end{align}
because the equality holds for all $\hat\rho$. This expression is actually condition $(\ref{item:probrel})$,
\begin{align}
     \pa_j &= \sum_i P_{ji}\ \pb_i,\  \forall\hat\rho.
\end{align}
This condition means there is a Markov chain $I\to J$ that fulfills condition $(\ref{item:ItoJ})$. Indeed, this chain has $P$ as transition matrix, and since in general
\begin{equation}
    p_{ij} = P_{ji}p_i \implies p_j = \sum_i P_{ji}p_i,
\end{equation}
if $I$ is distributed as $\pb_i$ then $J$ is distributed as $\pa_j$, $\forall\hat\rho$, completing the proof.
\end{proof}
This lemma characterizes coarser measurements using Markov chains, which modeled data processing in the previous section. Notice that only the marginals of $I$ and $J$ need to equate the distributions of $\mb$ and $\ma$. This relaxed requirement is crucial, because the joint distribution of performing $\mb$ and $\ma$ might not even be defined. Therefore, $\ma\coarser\mb$ does not mean that $\ma$ is necessarily obtained from processing $\mb$ outcomes. The lemma only entails that processing $\mb$ can generate in principle a measurement distributed as $\ma$. Also, notice that the transition matrix $P$ stays the same for all $\hat\rho$. 

Now, the central properties of coarser measurements appear from combining the processing inequalities in section \ref{sec:background} with the the relation between processing and coarser. 

\begin{theorem} Coarser measurements have a greater observational entropy for all $\hat\rho$,
\begin{equation}
    \ma \coarser \mb \implies S_{\ma}\geq S_{\mb}.
\end{equation}
The equality is achieved for all $\hat\rho$ if and only if \mbox{$\ma \finer \mb$}.\label{thm:CoarserObservational}
\end{theorem}
This result, proven in appendix \ref{apx:CoarserObservational}, is consistent with $S_\m$ measuring the entropy according to an observer that can only perform $\m$. An observer that has access to the finer measurement $\mb$ also has access to the coarser measurement $\ma$. Then, adequately, $S_{\mb}\leq S_{\ma}$.

The second property of coarser measurements is about non-increasing mutual information. Given a measurement $\m$, and $\hat\rho = \sum_x p_x \ketbra{x}{x}$, define 
\begin{equation}
    p_{xi} = \braket{x|\po_i|x} \braket{x|\hat\rho|x} = p_{i|x} p_x. \label{eq:jointdef}
\end{equation}
Computing the mutual information of this joint distribution is meaningful, because $I(p_{xi})$ quantifies how much $p_i$ tells us about $p_x$. These $p_x$ characterize the state of the system $\hat\rho$, and the $p_i$ are the distribution of $\m$ outcomes,
\begin{equation}
    p_i = \sum_x p_{xi} =\sum_x \braket{x|\hat\rho|x} \braket{x|\po_i|x}  = \Tr[\hat\rho \po_i].
\end{equation}
Therefore, $I(p_{xi})$ measures the information that $\m$ provides about $\hat\rho$, allowing to state the following,

\begin{theorem} Coarser measurements extract less information about $\hat\rho$,
\begin{equation}
    \ma \coarser \mb \implies I(\pa_{x j})\leq I(\pb_{x i}). \label{eq:lessinfo}
\end{equation}
The equality is achieved for all $\hat\rho$ if and only if \mbox{$\ma \finer \mb$}.\label{thm:CoarserProcessing}
\end{theorem}
This result, proven in appendix \ref{apx:CoarserProcessing}, confirms the intuition that coarser measurements are less capable to extract information about the state. All the information extracted by the coarser measurement is already in the finer measurement statistics.

\textit{\textbf{Measurements that are coarser in a subspace.-}}
Consider two measurement apparatus with their highest sensitivity on different ranges of values. This situation could be described as one of the measurements being coarser in a certain subspace of the Hilbert space, and the other measurement being coarser in a different subspace. The notion of coarser in subspaces is formalized in the current section, specifically for the case in which observers ignore the state of the system before measuring it. First, some notation is needed.
\begin{definition} The possible outcomes of $\m$ in a subspace $\mathcal{G}$ are denoted by $\om(\mathcal{G})$,
\begin{equation}
    \om(\mathcal{G}) = \{i \ |\ \exists \ket{\psi} \in \mathcal{G} \text{ s.t. } p_i = \bra{\psi}\po_i\ket{\psi} \neq 0 \}.
\end{equation}
For $\mb$ and $\ma$, the possible outcomes will be respectively denoted by $\ob(\mathcal{G})$ and $\oa(\mathcal{G})$.
\label{def:setoutcomes}
\end{definition}

\begin{definition} {\normalfont (Coarser in a subspace)} $\ma$ is coarser than $\mb$ in a subspace $\mathcal{G}\subseteq\mathcal{H}$, denoted by \mbox{$\ma \overset{\mathcal{G}}{\coarser} \mb$}, when
\begin{align}
      \hat P_{\mathcal{G}}\poa_j\hat P_{\mathcal{G}} &=\sum_{i\in \ob(\mathcal{G})} P_{ji}\hat P_{\mathcal{G}}\pob_i \hat P_{\mathcal{G}}, \  \forall j\in \oa(\mathcal{G}),\label{eq:poCond}\\
      \Va_j &\geq \sum_{i\in \ob(\mathcal{G})} P_{ji} \Vb_i, \  \forall j\in \oa(\mathcal{G}),\label{eq:volCond}
\end{align}
with $P$ a left stochastic $|\oa(\mathcal{G})|\times|\ob(\mathcal{G})|$ matrix and $\hat P_{\mathcal{G}}$ the projector onto $\mathcal{G}$.
\end{definition}
These conditions are less restrictive than demanding $\ma\coarser\mb$, therefore allowing to compare more measurements. Moreover, $\ma\pcoarser{\mathcal{G}}\mb$ reduces to $\ma\coarser\mb$ when $\mathcal{G}$ is the whole Hilbert space $\mathcal{H}$. The condition over the volumes, Eq.~\eqref{eq:volCond}, is the one reflecting that observers completely ignore the state of the system. Oppositely, this definition does not apply when the observer knows that the states live in $\mathcal{G}$. Given this restriction, the observer would instead define the volumes as $V_{\mathcal{G},i} = \Tr[\hat P_\mathcal{G}\po_i]$, because their guess of the state of the system would be $\hat P_\mathcal{G}/\Tr[\hat P_\mathcal{G}]$. Then, applying the condition Eq.~\eqref{eq:volCond} to $V_{\mathcal{G},i}$ would become unnecessary. This condition would always be fulfilled due to the relations among the POVM elements demanded in Eq.~\eqref{eq:poCond}. 

Interestingly, the definition of coarser in a subspace can also be derived from data processing, as shown next.
\begin{lemma} {\normalfont(Being coarser in a subspace is characterized by processing)} Consider $\mb$, $\ma$ with outcome probabilities $\pb_i$, $\pa_j$. Then, all the following are equivalent,
\begin{enumerate}[(i)]
    \item $\ma\pcoarser{\mathcal{G}}\mb$.\label{item:coarsersub}
    \item There is a $P$ left stochastic such that 
    \begin{align}
        \pa_j &= \sum_i P_{ji}\pb_i,\  \forall\hat\rho\in\mathcal{D}(\mathcal{G}),\label{eq:probrelsub}\\
        \Va_j &= \sum_i P_{ji}\Vb_i, \label{eq:volrelsub}
    \end{align} \label{item:probrelsub}
    where $\mathcal{D}(\mathcal{G})$ denotes the $\hat\rho$ formed from states in $\mathcal{G}$. 
    \item There exists a Markov chain $I\to J$ such that if $I$ is distributed as $\pb_i$ then $J$ is distributed as $\pa_j$, for all $\hat\rho\in\mathcal{D}(\mathcal{G})$ and for $\hat\rho^{\text{id}}$.\label{item:ItoJsub}
\end{enumerate} 
\label{lemma:processingmakescoarsersub}
\end{lemma}
The proof is in appendix \ref{apx:processingmakescoarsersub}. The connection between $\ma\pcoarser{\mathcal{G}}\mb$ and processing allows once more to use the inequalities of section \ref{sec:background} to show that $S_\m$ and mutual information are monotonic. Moreover, the previous theorems for coarser measurements are a particular case of the ones presented next, the case $\mathcal{G} = \mathcal{H}$. The two following theorems are respectively proven in appendices \ref{apx:CoarserObservationalSub} and \ref{apx:CoarserProcessingSub}.

\begin{theorem} Measurements that are coarser in $\mathcal{G}$ have a greater observational entropy in $\mathcal{G}$,
\begin{equation}
    \ma \pcoarser{\mathcal{G}} \mb \implies S_{\ma}(\hat\rho)\geq S_{\mb}(\hat\rho) ,\  \forall\hat\rho\in\mathcal{D}(\mathcal{G}).
\end{equation}
The equality occurs for all $\hat\rho\in\mathcal{D}(\mathcal{G})$ if and only if \mbox{$\ma \pfiner{\mathcal{G}} \mb$}.\label{thm:CoarserObservationalSub}
\end{theorem}

The next result uses the same joint distributions described in Eq.~\eqref{eq:jointdef}. 
\begin{theorem} Measurements that are coarser in $\mathcal{G}$ extract less information about the $\hat\rho$ in $\mathcal{D}(\mathcal{G})$,
\begin{equation}
    \ma \pcoarser{\mathcal{G}} \mb \implies I(\pa_{x j})\leq I(\pb_{x i}),\  \forall\hat\rho\in\mathcal{D}(\mathcal{G}). \label{eq:lessinfosub}
\end{equation}
The equality occurs $\forall\hat\rho\in\mathcal{D}(\mathcal{G})$ if and only if 
\begin{equation}
    \hat P_{\mathcal{G}}\poa_i\hat P_{\mathcal{G}} =\sum_{j\in \oa(\mathcal{G})} P_{ji}\hat P_{\mathcal{G}}\pob_j \hat P_{\mathcal{G}} ,\  \forall i\in \ob(\mathcal{G}),
\end{equation}
for a $P$ that is left stochastic.
\label{thm:CoarserProcessingSub}
\end{theorem}
Interestingly, Eq.~\eqref{eq:poCond} is enough to obtain that $I(\pa_{x j})\leq I(\pb_{x i}),\ \forall\hat\rho\in\mathcal{D}(\mathcal{G})$. The condition over the volumes, Eq.~\eqref{eq:volCond}, is not needed for this result.

Besides the monotonicity of observational entropy and mutual information, the definition of coarser in subspaces has the following desirable property.
\begin{theorem} {\normalfont (Coarseness is preserved when reducing the subspace)} Given $\mathcal{F} \subseteq\mathcal{G}$, 
\begin{equation}
     \ma\overset{\mathcal{G}}{\coarser}\mb \implies \ma\overset{\mathcal{F}}{\coarser}\mb. \label{eq:reduce}
\end{equation}
Moreover, a left stochastic matrix $P^{\mathcal{F}}$ for \mbox{$\ma\overset{\mathcal{F}}{\coarser}\mb$} can be obtained from restricting the matrix $P^{\mathcal{G}}$ of $\ma\overset{\mathcal{G}}{\coarser}\mb$ to the indices $\oa(\mathcal{F})\subseteq \oa(\mathcal{G})$ and  $\ob(\mathcal{F})\subseteq \ob(\mathcal{G})$.
\label{thm:preserved}
\end{theorem}

The implication is proven by applying lemma \ref{lemma:processingmakescoarsersub} twice, while considering $\mathcal{D}(\mathcal{F})\subseteq \mathcal{D}(\mathcal{G})$. The observations about $P^{\mathcal{F}}$ are proven in appendix \ref{apx:observations}. Coarser measurements are still coarser when the subspace is reduced, however, the opposite is false in general. The subspace where a measurement is coarser cannot be enlarged in general. Furthermore, coarseness is not preserved under the sum of subspaces, 
\begin{equation}
    \ma\overset{\mathcal{F}}{\coarser}\mb,\, \ma\overset{\mathcal{G}}{\coarser}\mb \centernot \implies \ma\pcoarser{\mathcal{F}+\mathcal{G}}\mb,
\end{equation}
as seen from the counterexample provided in appendix \ref{apx:counterexempleSum}. 
Furthermore, consider a measurement that is coarser in $\mathcal{F}$ and in $\mathcal{G}$, with $\mathcal{F}\subseteq\mathcal{G}$. Even in this situation, extending the $P^{\mathcal{F}}$ to $P^{\mathcal{G}}$ is not possible in general, as shown in appendix \ref{apx:counterexempleExtension}. Therefore, even if an extension of $P^{\mathcal{F}}$ to $P^{\mathcal{G}}$ does not exist, nothing can be said about whether $\ma\pcoarser{\mathcal{G}}\mb$.

Finally, let us go back to the starting example about measurement apparatus which have their best sensitivity on different scales. Appropriately, there are $\mb$ and $\ma$ where one is coarser on a subspace $\mathcal{F}$ and the other is coarser in a different subspace $\mathcal{G}$. For instance,
\begin{align}
    &\pob_1 = \ketbra{2}{2}+\ketbra{3}{3}, \ \pob_2 = \ketbra{0}{0}, \ \pob_3 = \ketbra{1}{1},  \nonumber \\
    &\poa_1 = \ketbra{0}{0}+\ketbra{1}{1}, \ \poa_2 = \ketbra{2}{2}, \ \poa_3 = \ketbra{3}{3},
\end{align}
with $\mathcal{F} = \Span(\ket{0},\ket{1})$ and $\mathcal{G} = \Span(\ket{2},\ket{3})$. In this example, only $\ma$ is coarser in $\mathcal{F}$ and only $\mb$ is coarser in $\mathcal{G}$.


\section{Recovering observational entropy properties using coarseness}
\label{sec:obsEntropyProp}

Several known important properties of observational entropy can be recovered quickly using coarser measurements. These properties support that $S_{\m}$ is a good measure of entropy when only the measurement $\m$ is available.  

\begin{definition} {\normalfont (Measurement associated to a self-adjoint operator $\hat A$)} The measurement $\mathcal C_{\hat A}$ consists of the projectors $\hat P_a$ of the eigendecomposition $\hat A = \sum_a a \hat P_a$.
\end{definition}
\begin{theorem} {\normalfont(Bounds on the observational entropy)}
\begin{equation}
    S_{\text{vN}}(\hat\rho) \leq S_{\m}(\hat\rho) \leq \ln \dim \mathcal{H}, \label{eq:bounds}
\end{equation}
with the equality cases
\begin{align}
    S_{\text{vN}}(\hat\rho) = S_{\m}(\hat\rho) &\iff \m_{\hat\rho}\coarser\m, \label{eq:equalVN}\\
    S_{\m}(\hat\rho) = \ln \dim \mathcal{H} &\iff p_i = \frac{V_i}{\dim \mathcal{H}}
    .\label{eq:equalId}
\end{align}
\label{thm:propObs}
\end{theorem}

\begin{proof}
The first step is to convert $\ln \dim \mathcal{H}$ and $S_{\text{vN}}(\hat\rho)$ into observational entropies, where the coarseness inequalities can be used. This step is accomplished realizing that
\begin{equation}
    S_{\text{vN}}(\hat\rho) = S_{\m_{\hat\rho}}(\hat\rho),\quad \ln \dim\mathcal{H} = S_{\m_{\hat 1}}(\hat\rho),  \label{eq:intermsObs}
\end{equation}
where the $\hat 1$ in $\m_{\hat 1}$ is the identity operator.

Next, let us $S_{\m_{\hat\rho}}(\hat\rho) \leq S_{\m}(\hat\rho)$. Start decomposing \mbox{$\hat\rho = \sum_\rho \rho \hat P_{\rho}$}, then 
\begin{align}
    p_i=\sum_\rho P_{i\rho}p^{(\hat\rho)}_\rho,\quad  V_i=\sum_\rho P_{i\rho}V^{(\hat\rho)}_\rho,\label{eq:coarserpoint}
\end{align}
where $P_{i\rho} = \Tr\{\po_i\hat P_{\rho}\}/V_\rho^{(\hat\rho)}$. These identities imply $S_{\m}(\hat\rho)\leq S_{\m_{\hat\rho}}(\hat\rho)$, due to theorem \ref{thm:observationalmonotonous}. The equality case requires a longer argument provided in appendix \ref{apx:equality}. 

Lastly, to prove $S_{\m}(\hat\rho) \leq S_{\m_{\hat 1}}(\hat\rho)$, check
\begin{equation}
    \m_{\hat 1}\coarser\m, \text{ with } P_{1i} = 1
\end{equation} 
and use theorem \ref{thm:CoarserObservational}.
Moreover, the probabilities for $\hat\rho$ and for $\hat\rho^{\text{id}}$ fulfill the hypotheses of theorem \ref{thm:observationalmonotonous}, implying that $S_{\m}(\hat\rho) = S_{\m_{\hat 1}}(\hat\rho)$ happens if and only if 
\begin{equation}
    P_{ji}p_{i}V_j=P_{ji}V_ip_j.
\end{equation}
This equation becomes the required $p_{i}\dim \mathcal{H}=V_i$ after substituting the known values.
\end{proof}

The equivalence in Eq.~\eqref{eq:equalVN} is remarkably useful, allowing to quickly check whether $\ma \coarser \mb$ for any projective $\ma$. The check is based in using any state $\hat\rho_2$ with eigenspaces defined by the projectors of $\ma$, equivalently, \mbox{$\m_{\hat\rho_2}=\ma$}. Then, computing whether \mbox{$S_{\text{vN}}(\hat\rho_2)=S_{\mb}(\hat\rho_2) $} is enough to learn whether $\ma \coarser \mb$.

The remaining property of observational entropy is about consecutive measurements. Consider an observer with the state $\hat\rho$ that performs first $\mb$ and then $\ma$. This new measurement is denoted by $(\mb,\ma)$ and will have Kraus operators $\{\hat K^{(1)}_{in} \hat K^{(2)}_{jm}\}$. Then,
\begin{theorem} Performing additional measurements makes the measurement finer,
\begin{equation}
    (\mb,\ma)\finer\mb.
\end{equation}
Therefore, due to theorem \ref{thm:CoarserObservational}, performing additional measurements cannot increase observational entropy,
\begin{equation} 
    S_{(\mb,\ma)} \leq S_{\mb},
\end{equation}
and the equality is achieved if and only if $(\mb,\ma)\coarser\mb$.
\end{theorem}
\begin{proof}
Express the POVM elements in terms of the Kraus operators,
\begin{align}
    \pob_{i} &= \sum_n K_{in}^{(1) \dagger} K_{in}^{(1)}, \nonumber\\
    \po^{(1,2)}_{ij} & = \sum_{m,n}   K_{in}^{(1)\dagger} K_{jm}^{(2)\dagger} K_{jm}^{(2)}K_{in}^{(1)}  \nonumber \\
    &= \sum_n  K_{in}^{(1)\dagger} \poa_{j} K_{in}^{(1)}.
\end{align}
Then, using $\sum_j\pob_j = \id$,
\begin{equation}
    \pob_{i} = \sum_j \po^{(1,2)}_{ij},
\end{equation}
which proves $(\mb,\ma)\finer \mb$. This proof can be understood intuitively as obtaining $\mb$ from $(\mb,\ma)$ by throwing away the outcomes of $\ma$.
\end{proof}
Therefore, performing more measurements never increases the observational entropy, which successfully captures that performing more measurements never reduces the known information.

\section{Conclusion}
Performing coarser measurements provides less information. For projective measurements in quantum systems, the notion of coarseness is established. However, this notion is not as established for the most general type of measurements, POVM measurements. The present article proposes a definition of coarser that extends previous ones, applies to POVM measurements and more generally to any measurements beyond quantum mechanics. This definition is based on the idea that processing can only make measurements coarser. Specifically, observers can process the outcomes of their measurements, but doing so does not provide new information about the measured system. This connection can be translated to a relation between the POVM elements of coarser and finer measurements, given in Eq.~\eqref{eq:coarser}.

Two important measures of information are monotonic with how much coarse are the measurements. Observational entropy quantifies the unknown information when only one measurement is available and increases with coarser measurements. Conversely, mutual information quantifies the correlations between the state and the measurement outcomes and decreases with coarser measurements. These behaviours reflect that coarser measurements extract less information, supporting the presented definition of coarser. 

On the one hand, quantum information benefits from having a definition of coarser, allowing to compare different measurement apparatus. Especially, being coarser in subspaces of states describes measuring apparatus that have their best accuracy on different ranges. In the future, measurements in infinite dimensional systems could be explored, seeing if the properties of coarser measurements change in such systems. Another direction worth exploring is the connection between coarser measurements and causal models. This connection comes from the idea that the processing of measurement outcomes, which is used to define coarser, can be thought of as a causal model.

On the other hand, quantum thermodynamics is connected to coarser measurements through observational entropy. Observational entropy has recently been put forward as candidate for thermodynamic entropy, in the context of quantum thermodynamics \cite{PRXQuantum.2.030202,safranek2019thermRapid,brief2021,safranek2019therm,strasberg2019,classical2020}. The properties of observational entropy can be more easily derived using coarser measurements, as shown in section \ref{sec:obsEntropyProp}. Moreover, all the measures of information that fulfill data processing inequalities are monotonic with this definition of coarser. These measures include mutual information, Rényi entropies, relative entropy and other f-divergences. Moreover, this behaviour is reminiscent of the second law of thermodynamics, where entropy is monotonic. Altogether, the coarseness of measurements becomes a tool to characterize the behaviour of information and, in turn, develop the thermodynamics of classical and quantum systems.

\acknowledgments
The author gratefully acknowledges José Polo-Gómez for his helpful comments on the final version of this paper. The author also thanks Leonardo Lessa, Eduardo Martín-Martínez, José Polo-Gómez and Tales Rick Perche for insightful discussions. The project that lead to the presented results received the support of a fellowship from ``la Caixa” Foundation (ID 100010434, with fellowship code LCF/BQ/EU21/11890119) and of Eduardo Martín-Martínez's funding through the NSERC Discovery program as well as his Ontario Early Researcher Award.

\newpage 

\appendix

\section{Deriving Properties of Observational Entropy for Generalized Measurements}

\subsection{Failure of the Relation between Observational and von Neumann Entropies}
\label{apx:failedvNrel}
This appendix subsection shows that Eq.~\eqref{eq:interpret}, while valid for projective measurements, is not valid for generalized measurements. Using the same equation, but replacing the $\hat P_i$ by general POVM elements $\po_i$ does not work. Moreover, one could expect to at least maintain the information-theoretic interpretation of observational entropy previously deduced from this equation, but that is not the case either. These claims are proven using as counterexample
\begin{align}
    &\po_1= \frac{1}{2}\ketbra{0}{0},\ \po_2 = \frac{1}{2}\ketbra{0}{0} + \ketbra{1}{1}, \ \hat \rho  = \ketbra{0}{0},\nonumber\\
    &\text{where}\quad p_1 = \frac{1}{2},\ V_1=\frac{1}{2},\ p_2 = \frac{1}{2},\ V_2=\frac{3}{2}. \label{eq:counterapxA}
\end{align}

First, assuming that $\hat P_i$ had to be replaced by $\po_i$, Eq.~\eqref{eq:interpret} would become
\begin{align}
    S_{\m}(\hat\rho) &= S_{\text{vN}}\left(\sum_{i} p_i \frac{\po_i}{V_i}\right)\\
    \frac{1}{2}\ln\left(\frac{3}{2}\right) &= S_{\text{vN}}\left( \frac{2}{3}\ketbra{0}{0} + \frac{1}{3}\ketbra{1}{1} \right)\\
    \frac{1}{2}\ln\left(\frac{3}{2}\right) &= \frac{2}{3}\ln\left(\frac{3}{2}\right) + \frac{1}{3}\ln 3 \\
    \ln 2 &= 3\ln 3,
\end{align}
which is false. Therefore, Eq.~\eqref{eq:interpret} cannot be extended to generalized measurements $\po_i$ through replacing the $\hat P_i$.

Second, the information-theoretic interpretation that came from Eq.~\eqref{eq:interpret} also fails. This interpretation can be expressed as
\begin{equation}
    S_{\m}(\hat\rho) = S_{\text{vN}}(\hat\rho_{\text{est}}),\label{eq:interpfail}
\end{equation}
where $\hat\rho_{\text{est}}$ is the best estimation of $\hat\rho$ given that the $p_i$ and $\po_i$ are known. Interestingly, in general
\begin{equation}
    \hat\rho_{\text{est}} \neq \sum_{i} p_i \frac{\po_i}{V_i},
\end{equation}
even though the equality is true when the $\po_i$ are projectors. This failure of the equality happens for the counterexample in Eq.~\ref{eq:counterapxA}, for which only $\hat\rho=\ketbra{0}{0}$ fulfills $p_1=1/2$. Therefore, $\hat\rho_{\text{est}} = \ketbra{0}{0}$, which is different from
\begin{equation}
    \sum_{i} p_i \frac{\po_i}{V_i} = \frac{2}{3}\ketbra{0}{0} + \frac{1}{3}\ketbra{1}{1},
\end{equation}
as we wanted to see. Moreover, $\hat\rho_{\text{est}} = \ketbra{0}{0}$ also implies that the interpretation in  Eq.~\eqref{eq:interpfail} fails because $ S_{\m}(\hat\rho)\neq0$, while $S_{\text{vN}}(\hat\rho_{\text{est}})=0$. In words, the observational entropy is non-zero, even though the $p_i$ tell us everything about $\hat\rho$ in the present counterexample. These results justify the need to modify the information-theoretic interpretation for $S_{\m}$, as done in section \ref{sec:background}.

\subsection{Proof of Theorem \ref{thm:observationalmonotonous}}
\label{apx:dpiobs}
\begin{proof}
Consider $I$ and $J$ such that $I\to J$ with transition matrix $P$. If $p_i$ and $p_i^{\text{id}}$ are distributions of $I$, then $p_j$ and $p_j^{\text{id}}$ are the corresponding distributions of $J$. This claim is true due to Eq.~\eqref{eq:probrel}, together with $V_i/V_{\text{tot}}=p_i^{\text{id}}$, $p_{j}=\sum_i p_{ij}$ and $p_{j}^{\text{id}}=\sum_i p^{\text{id}}_{ij}$. Then, using $I\to J$, the strong data processing inequality (theorem \ref{thm:dataProcStr}) implies that $S_{\text{obs}}$ is non-decreasing, due to
\begin{equation}
    S_{\text{obs}}(p_i,V_i) = \ln V_\text{tot} - D_{\text{KL}}[p_i||p_i^{\text{id}}].
\end{equation}
Moreover, theorem \ref{thm:dataProcStr} also implies that $S_{\text{obs}}$ remains constant if and only if 
\begin{align}
    p_{ij}=p^{\text{id}}_{i|j}p_j.
\end{align}
Equivalently, using that $p^{\text{id}}_{i|j}V_j=p^{\text{id}}_{j|i}V_i$ and that $p_{j|i}=p^{\text{id}}_{j|i}$ and assuming that the $V_j$ are non-zero,
\begin{equation}
    p_{j|i}p_iV_j=p_{j|i}V_ip_j,
\end{equation}
which is the condition given for the equality case, proving the theorem for finite $V_\text{tot}$. The theorem holds even if $V_\text{tot}$ is infinite, as shown from Jensen's inequality in the following appendix \ref{apx:InfiniteVtot}.
\end{proof}

\subsection{Proving Theorem \ref{thm:observationalmonotonous} for Systems of Infinite Volume}
\label{apx:InfiniteVtot}
The theorem \ref{thm:observationalmonotonous} for finite $V_\text{tot}$ is proven above, but a derivation that is valid even for infinite $V_{\text{tot}}$ is provided next. This derivation is mainly based on the proof of theorem 3 in \cite{safranek2021information}, where Jensen's inequality is also central.
\begin{theorem} {\normalfont (Jensen's inequality)} Let $f$ be a strictly concave function, consider any weights $a_i$ , $0 \leq a_i \leq 1$, $\sum_i a_i = 1$, and any values $x_i\in \dom(f)$. Then,
\begin{equation}
    f\left(\sum_i a_i x_i\right) \geq \sum_i a_i f(x_i).
\end{equation}
The equality is true if and only if $x_i = x$, $\forall i\, |\, a_i\neq 0$, where $x$ is a constant value.
\label{thm:Jensen}
\end{theorem}

Then, the proof of theorem \ref{thm:observationalmonotonous} goes as follows. Consider the hypotheses
\begin{align}
     p_j &= \sum_i P_{ji} p_i,\  V_j = \sum_i P_{ji} V_i, \label{eq:hypotheses}
\end{align}
for a $P$ left stochastic. Moreover, let $f(x) = -x\ln x$, which is strictly concave and $\dom(f)=[0,\infty)$. Then,
\begin{align}
    &\!\!S_{\text{obs}}(p_j,V_j) \nonumber\\&= -\sum_j p_j\ln\frac{p_j}{V_j}\nonumber\\
    &= \sum_j V_j f\left(\frac{p_j}{V_j}\right)\nonumber\\
    &= \sum_j V_j f\left(\sum_i \frac{P_{ji}V_i}{V_j} \frac{p_i}{V_i}\right)\nonumber\\
    &\geq \sum_j V_j \sum_i \frac{P_{ji}V_i}{V_j} f\left(\frac{p_i}{V_i}\right)\nonumber\\
    &= \sum_i V_i f\left(\frac{p_i}{V_i}\right) \nonumber\\
    &=  S_{\text{obs}}(p_j,V_j),
\end{align}
due to using the hypotheses and using Jensen's inequality with
\begin{equation}
    a_{ji} = \frac{P_{ji}V_i}{V_j},\quad x_i = \frac{p_i}{V_i},
\end{equation}
which fulfill the requirements, stated in theorem \ref{thm:Jensen}. Therefore, 
\begin{align}
    &p_j = \sum_i P_{ji} p_i,\ V_j = \sum_i P_{ji} V_i\nonumber\\
    &\implies S_{\text{obs}}(p_j,V_j)\geq S_{\text{obs}}(p_i,V_i). \label{eq:monotony}
\end{align}
Moreover, the equality is achieved if and only if
\begin{equation}
    x_i = c_j ,\  \forall i\, |\,  a_{ji} \neq 0, \label{eq:equalcond}
\end{equation} 
where $c_j$ is a constant independent of $i$. Equivalently, $a_{ji}c_j = a_{ji}x_i$, which allows to compute $c_j$, 
\begin{equation}
    c_j = \sum_i a_{ji} c_j = \sum_i \frac{P_{ji}V_i}{V_j} x_i = \sum_i \frac{P_{ji}p_i}{V_j} = \frac{p_j}{V_j}.
\end{equation}
Substituting everything into $a_{ji}c_j = a_{ji}x_i$ and assuming that the volumes are non-zero, then
\begin{align}
    &p_j = \sum_i P_{ji} p_i,\ V_j = \sum_i P_{ji} V_i, \ P_{ji}p_i V_j = P_{ji} p_j V_i  \nonumber \\ &\iff S_{\text{obs}}(p_j,V_j)=S_{\text{obs}}(p_i,V_i),
\end{align}
completing the proof.

\subsection{Counterexample for the Converse of Theorem \ref{thm:observationalmonotonous}}
\label{apx:counterexampleConverse}
The main result of theorem \ref{thm:observationalmonotonous} is an implication whose converse is false, i.e.
\begin{align}
    &\pa_j = \sum_i P_{ji} \pb_i, 
    \ \Va_j = \sum_i P_{ji} \Vb_i \nonumber \\
    \centernot \impliedby & S_{\text{obs}}(\pVa) \geq S_{\text{obs}}(\pVb). \label{eq:falseconverse}
\end{align} 
There is a simple counterexample, already for measurements of two outcomes, which proves that this converse is false, 
\begin{align}
    &\pb_1 = \frac{3}{4},\ \Vb_1 = 1. \nonumber \\
    &\pa_1 = 1,\ \Va_1 = \frac{9}{5}, \label{eq:counterexample}
\end{align}
Here, only output $1$ appears, because the probabilities and volumes of the output $2$ are determined by the probabilities adding to one and the volumes adding to $V_\text{tot}=2$. Next, computing $S_{\text{obs}}$, 
\begin{equation}
    0.6 \approx S_{\text{obs}}(\pa_j,\Va_j) >  S_{\text{obs}}(\pb_i,\Vb_i) \approx 0.2.
\end{equation}
However, 
\begin{align}
     \pa_j &= \sum_i P_{ji} \pb_i,\ \Va_j = \sum_i P_{ji} \Vb_i
\end{align}
cannot be fulfilled. Arguing by contradiction, assume that an adequate $P_{ji}$ can be found. Then, $\pa_1 = 1$ needs that $\pa_1 = \pb_1 + \pb_2$, which implies $1.8 =\Va_1 = \Vb_1 + \Vb_2 = 2$, reaching a contradiction. In general, for any given $(\pb_1, \Vb_1)$, there are many more $(\pa_1, \Va_1)$ that can serve as counterexample. These counterexamples are visualized in the Figure \ref{fig:regions} and any of them is enough to show Eq.~\eqref{eq:falseconverse}.

\begin{figure}[ht]
    \centering
    \includegraphics[scale=0.42]{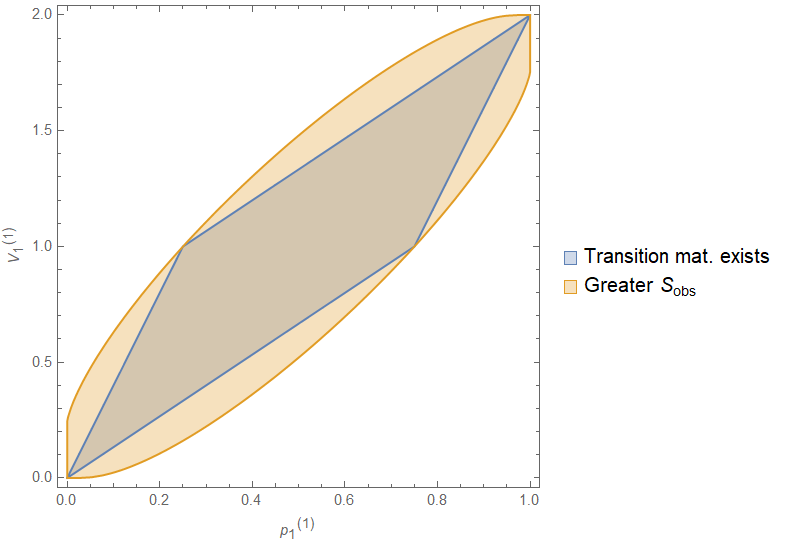}
    \label{fig:regions}
    \caption{Given $\pb_1 = 0.75,\ \Vb_1 = 1$, regions of the $\pa_1,\Va_1$ square for which observational entropy is greater (orange) and for which the conditions in Eq.~\eqref{eq:probrel} can be fulfilled (blue). Notably, these regions are different, allowing to find counterexamples to the converse of theorem \ref{thm:observationalmonotonous} main implication.}
\end{figure}

Moreover, the counterexample described in Eq.~\eqref{eq:counterexample} can arise from generalized measurements on a two-dimensional quantum system. The adequate probabilities and volumes are obtained from the state
\begin{equation}
    \ket{\psi} = \frac{\sqrt{3}}{2}\ket{0} + \frac{1}{2}\ket{1},\quad \hat\rho = \ket{\psi}\bra{\psi},
\end{equation}
together with the measurements
\begin{align}
    \poa_1 &= \ket{\psi}\!\bra{\psi}+\frac{4}{5}\ket{\psi^\perp}\!\bra{\psi^\perp}, \ \poa_2 = \frac{1}{5}\ket{\psi^\perp}\!\bra{\psi^\perp},\nonumber\\
    \pob_1 &= \ket{0}\!\bra{0},\ \pob_2 = \ket{1}\!\bra{1},
\end{align}
where $\ket{\psi^\perp}$ is a state orthogonal to $\ket{\psi}$.

\subsection{Proving when Observational and von Neumann Entropies Coincide}
\label{apx:equality}

This section completes the proof of theorem \ref{thm:propObs}, by demonstrating Eq.~\eqref{eq:equalVN}. The already proven Eq.~\eqref{eq:coarserpoint} allowed to apply theorem \eqref{thm:observationalmonotonous} to show $S_{\m_{\hat\rho}}(\hat\rho) \leq S_{\m}(\hat\rho)$ in the main text. The same theorem states that the equality $S_{\m_{\hat\rho}}(\hat\rho)= S_{\m}(\hat\rho)$ occurs when
\begin{equation}
    \frac{p_\rho^{(\hat\rho)}}{V_\rho^{(\hat\rho)}} =\frac{p_i}{V_i} ,\  \forall P_{i\rho} \neq 0 \label{eq:condequal}.
\end{equation}
This condition allows to prove Eq.~$\eqref{eq:equalVN}$ in a similar way to part of the proof of theorem 1 in the reference \cite{safranek2021information}. Consider the sets
\begin{equation}
    I_{\rho} = \{i | P_{i\rho} \neq 0 \},
\end{equation}
which are disjoint because 
\begin{equation}
    \exists i \in I_{\rho}\cap I_{\rho'} \implies \rho = \frac{p_\rho^{(\hat\rho)}}{V_\rho^{(\hat\rho)}} =  \frac{p_i}{V_i} = \frac{p_{\rho'}^{(\hat\rho)}}{V_{\rho'}^{(\hat\rho)}} = \rho'.
\end{equation}
If $i\in I_{\rho}$, then $P_{i\rho'} = 0$, $\forall\rho'\neq\rho$, which implies $\Tr\{\po_i\hat P_{\rho'}\}=0$, and using the lemma \ref{eq:trequality}, $\po_i\hat P_{\rho'}=0$. Consequently,
\begin{equation}
    \po_i = \po_i\sum_{\rho''} \hat P_{\rho''} = \po_i \hat P_\rho.
\end{equation}
Moreover, combining $\po_i\hat P_{\rho'}=0$ and $\po_i = \po_i \hat P_\rho$,
\begin{equation}
    \hat P_\rho = \sum_i \po_i \hat P_{\rho} = \sum_{i\in I_{\rho}} \po_i \hat P_{\rho} = \sum_{i\in I_{\rho}} \po_i,
\end{equation}
which means $\m_{\hat\rho}\coarser\m$, and in turn, $S_{\m_{\hat\rho}}(\hat\rho) \geq S_{\m}(\hat\rho)$ by theorem \ref{thm:CoarserObservational}. Combining all the results with \mbox{$S_{\text{vN}}(\hat\rho) = S_{\m_{\hat\rho}}(\hat\rho)$} from Eq.~\eqref{eq:intermsObs} provides the final outcome
\begin{equation}
    S_{\text{vN}}(\hat\rho) = S_{\m}(\hat\rho) \iff \m_{\hat\rho}\coarser\m.
\end{equation}

\section{Proving the Data Processing Inequalities}
\label{apx:dpi}

\subsection{Proof of Theorem \ref{thm:dataProcStr}}
\label{apx:dpistr}

\begin{proof} The proof consists of using $I\to J$ first and Gibbs' inequality second, to derive
\begin{align}
    D_{\text{KL}}[p_i||q_i] &= D_{\text{KL}}[p_{ij}||q_{ij}] \nonumber \\
    &= \sum_{ij} p_{ij} \ln\frac{p_{i|j}p_j}{q_{i|j}q_j} \nonumber \\ 
    &= D_{\text{KL}}[p_j||q_j] + \sum_{i} p_j D_{\text{KL}}[p_{i|j}||q_{i|j}] \nonumber \\ 
    &\geq D_{\text{KL}}[p_j||q_j].
\end{align}
The equality in the first line comes from $p_{j|i} = q_{j|i}$, which is true due to $I\to J$. The last line comes from $D_{\text{KL}}[p_{i|j}||q_{i|j}] \geq 0$ and the equality is achieved if and only if $D_{\text{KL}}[p_{i|j}||q_{i|j}] = 0$, $\forall p_j \neq 0$. This condition can be restated as $p_{i|j} = q_{i|j}$, \mbox{$\forall p_j \neq 0$} by using Eq.~\eqref{eq:gibbs}. Then, multiplying both sides of the condition by $p_j$, \mbox{$p_{ij} = q_{i|j}p_j$}, completing the proof. 
\end{proof}

\subsection{Proof of Theorem \ref{thm:DPIMI}}
\label{apx:dpimi}
\begin{proof}
Using the relations
\begin{align}
    I(p_{xy}) &= \sum_x p_x D_{\text{KL}}(p_{y|x}||p_y),\\
    I(p_{xz}) &= \sum_x p_x D_{\text{KL}}(p_{z|x}||p_z),
\end{align}
reduces proving $I(p_{xy}) \geq I(p_{xz})$ to proving
\begin{equation}
    D_{\text{KL}}(p_{y|x}||p_y) \geq D_{\text{KL}}(p_{z|x}||p_z),\  \forall x.\label{eq:ineqcond}
\end{equation}
These inequalities hold because of the strong data processing inequality in theorem \ref{thm:dataProcStr}. The hypotheses of this theorem are fulfilled because $X\to Y \to Z$ implies
\begin{equation}
    p_{yz} = p_{z|y} p_y, \ p_{yz|x} = p_{z|y} p_{y|x}.
\end{equation}
Therefore, both $p_y$ and $p_{y|x}$ are evolved by the same transition matrix $p_{z|y}$, which allows to apply theorem \ref{thm:dataProcStr} to prove Eq.~\eqref{eq:ineqcond}. In turn, $I(p_{xy}) \geq I(p_{xz})$ and the equality is obtained if and only if
\begin{equation}
    p_{yz|x} = p_{y|z}p_{z|x},\  \forall p_x\neq 0,
\end{equation}
where $p_{y|z}$ is defined by $p_{zy}= p_{y|z}p_z$. Equivalently, multiplying by $p_x$ on both sides,
\begin{equation}
    p_{xyz} = p_{y|z}p_{xz},
\end{equation}
completing the proof.
\end{proof}

\section{Simplifying the Definition of Coarser when Measurements are Projective}
\label{apx:generalized}
This section proves theorem \ref{thm:projective}, which implies that the new definition of coarser reduces to the earlier definition when applied to projective measurements. The following lemma will be helpful to complete the proof.
\begin{lemma} Given a projector $\hat P$ and a positive operator $\po$,
\begin{equation}
    \Tr[\po\hat P] \geq 0 
    \label{eq:trpositive}
\end{equation}
and the equality case
\begin{equation}
    \Tr[\po\hat P] = 0 \iff  \hat P\po = \po\hat P = 0.
    \label{eq:trequality}
\end{equation}
\label{lemma:projandpo}
\end{lemma}
\begin{proof}
To start with, consider the decomposition of the projector into orthogonal eigenvectors $\hat P = \sum_x \ketbra{x}{x}$. Then,
\begin{equation}
    \Tr[\po\hat P] = \sum_x \braket{x|\po|x} \geq 0
\end{equation}
because $\po$ is positive, proving the inequality. 

In the equality case,
\begin{equation}
    \sum_x \braket{x|\po|x} = 0 \implies \braket{x|\po|x} = 0,
\end{equation}
again due to $\po$ being positive. Now, decompose $\po$ into Kraus operators $\hat K_l$,
\begin{equation}
    \po = \sum_l \hat K_l^\dagger \hat K_l.
\end{equation}
And substituting this decomposition into the previous equation,
\begin{align}
    0 = \braket{x|\po|x} = \braket{x|\sum_l\hat K_l^\dagger \hat K_l|x} = \sum_l \|\hat K_l\ket{x}\|^2,
\end{align}
which implies $\hat K_l\ket{x} = 0$. Therefore, $\hat K_l^\dagger \hat K_l\ketbra{x}{x} = 0$ and $\ketbra{x}{x}\hat K_l^\dagger \hat K_l = 0$, and adding over all $l$ and $x$, $\po\hat P = 0$ and $\hat P\po=0$. In summary, 
\begin{equation}
    \Tr[\po\hat P] = 0 \implies \hat P\po = \po\hat P = 0,
\end{equation}
and the converse is trivial, because $\Tr[0] = 0$.
\end{proof}

\subsection{Proof of Theorem \ref{thm:projective}}

\begin{proof} $\implies)$
Let $i$ be an outcome of $\mb$ and $j$ an outcome of $\ma$ such that $\pob_{i}\Pa_{j}\neq 0$. This $j$ must exist, otherwise,
\begin{equation}
    \pob_{i}=\pob_{i} \id = \pob_{i} \sum_{j'}\Pa_{j'} = 0,
\end{equation}
which is false. Next, consider the hypothesis $\ma \coarser \mb$, which implies
\begin{align}
    \Pa_{j'} &= \sum_{i'} P_{j'i'} \pob_{i'}.
    \label{eq:hypcoarser}
\end{align}
Combining this equation with $\Pa_{\bar j}\Pa_j = 0$ if $\bar j\neq j$, which is true because the projectors that form measurements are orthogonal, 
\begin{align}
    0 &= \Tr[\Pa_{\bar j}\Pa_j] \nonumber\\
    &= \sum_{i'} P_{\,\bar ji'} \Tr[\pob_{i'}\Pa_{j}] \nonumber\\
    &\geq  P_{\,\bar ji} \Tr[\pob_i\Pa_{j}].
    \label{eq:ineqPij}
\end{align}
Where the last line was obtained discarding non-negative terms, according to Eq.~\eqref{eq:trpositive} of the previous lemma and $P_{\,\bar ji'}\geq0$. The same lemma, together with $\pob_{i}\Pa_{j}\neq 0$, implies $\Tr[\pob_{i}\Pa_{j}]>0$. Therefore, $P_{\,\bar ji} = 0$ to fulfill the inequality in Eq.~\eqref{eq:ineqPij}. Since the value of $\bar j$ was not specified, $P_{\,\bar ji} = 0$, $\forall \bar j \neq j$. In other words, there is only one $j'$ for which $P_{j'i'} \neq 0$. Consequently, the sets
\begin{equation}
    I_{j'} = \{i'\,|\, P_{j'i'} \neq 0\}
\end{equation}
are disjoint. Moreover, $i\in I_{j}$ implies
\begin{equation}
    P_{ji} = \sum_{j'} P_{j'i} = 1,
\end{equation}
since $P_{\bar ji} = 0$, $\forall \bar j \neq j$. Then, the hypothesis $\ma \coarser \mb$ becomes
\begin{equation}
    \Pa_j = \sum_i P_{ji} \pob_i = \sum_{i\in I_{j}} P_{ji} \pob_i = \sum_{i\in I_{j}} \pob_i.
\end{equation}
Since $I_{j}$ are disjoint, the $\implies$ implication is proven. 

$\impliedby)$ Building a left stochastic $P$ for which $\ma \coarser \mb$ is enough to prove the converse. Let $P_{ji} = 1$ if $i\in I_{j}$ and $P_{ji}=0$ otherwise. Then, the proof would finish if for each $i$ existed $j$ such that $i\in I_{j}$, because
\begin{equation}
    \sum_{j'} P_{j'i} = P_{ji} = 1,
\end{equation}
as needed. To prove that such $j$ actually exists is equivalent to prove that $\bigcup_j I_{j}$ are all the outcomes of $\mb$. Consider
\begin{equation}
    \id = \sum_j \Pa_j = \sum_j \sum_{i\in I_{j}} \pob_i = \sum_{i\in \bigcup_j I_{j}} \pob_i,
\end{equation}
where the last equality is true thanks to the $I_{j}$ being disjoint. Finally, 
\begin{equation}
    0 = \id - \sum_{i\in \bigcup_j I_{j}} \pob_i = \sum_{i \not \in \bigcup_j I_{j}} \pob_i,
\end{equation}
where we used that $I_{j}$ are disjoint again and the completeness relation for the $\{\pob_i\}$. Since $\pob_i$ are positive, if $i \not \in \bigcup_j I_{j}$ then $\pob_i=0$, which is not allowed. Therefore all $i\in \bigcup_j I_{j}$ and from the previous reasoning, $P$ is left stochastic and $\ma \coarser \mb$. Therefore proving the remaining implication 
\begin{equation}
    \ma \coarser \mb \impliedby \Pa_j = \sum_{i\in I_{j}} \pob_i,
\end{equation}
with $I_{j}$ disjoint sets.
\end{proof}

\section{Proving the Properties of Coarser Measurements}
\label{apx:coarser}
Here are the proofs of section \ref{sec:CoarserAndProcessing} regarding the properties of coarser measurements.
\subsection{Proof of Theorem \ref{thm:CoarserObservational}}
\label{apx:CoarserObservational}
\begin{proof}
This result follows from theorem \ref{thm:observationalmonotonous}, whose hypothesis are fulfilled because of lemma \ref{lemma:processingmakescoarser}. According to this lemma, $\ma\coarser\mb$ is equivalent to
\begin{equation}
        \pa_j = \sum_i P_{ji}\ \pb_i,\  \forall\hat\rho.
\end{equation} 
In particular, taking the state $\hat\rho^{\text{id}}$, whose probabilities are proportional to the volumes,
\begin{equation}
        \Va_j = \sum_i P_{ji}\ \Vb_i. \label{eq:volrelproof}
\end{equation} 
Then, using the last two equations, theorem \ref{thm:observationalmonotonous} implies 
\begin{equation}
    S_{\ma}(\hat\rho)\geq S_{\mb}(\hat\rho),\  \forall\hat\rho,
\end{equation}
with equality for each $\hat\rho$ if and only if 
\begin{equation}
    P_{ji}\pb_i\Va_j=P_{ji}\Vb_i\pa_j,
\end{equation}
Then, assuming $\Va_j>0$,
\begin{equation}
    \pb_i = \sum_j P_{ji} \pb_i = \sum_j P_{ji}\frac{\Vb_i}{\Va_j} \pa_j = \sum_j \tilde P_{ij} \pa_j,\  \forall\hat\rho,\label{eq:reverse}
\end{equation}
where $\tilde P_{ij}$ is left stochastic, because
\begin{itemize}
    \item $\tilde P_{ij}\geq 0$, due to \mbox{$P_{ji},\Vb_i,\Va_j\geq0$}.
    \item $\sum_i \tilde P_{ij} =\sum_i P_{ji}\frac{\Vb_i}{\Va_j} = 1$, due to Eq.~\eqref{eq:volrelproof}.
\end{itemize}
Finally, Eq.~\eqref{eq:reverse} is equivalent to \mbox{$\mb\coarser\ma$}, again due to lemma \ref{lemma:processingmakescoarser}. Therefore, the equality in $S_{\ma}\geq S_{\mb}$ is achieved if and only if $\mb\coarser\ma$.
\end{proof}

\subsection{Proof of Theorem \ref{thm:CoarserProcessing}}
\label{apx:CoarserProcessing}
\begin{proof}
This proof consists of defining a Markov chain $X\to I \to J$ with adequate marginals $\pa_{xj}$ and $\pb_{xi}$ and then applying theorem \ref{thm:DPIMI}. First, define the joint probabilities of $X$, $I$, $J$ as
\begin{equation}
    p_{xij} = P_{ji}\pb_{xi} = P_{ji} \braket{x|\pob_i|x} p_x, \label{eq:chaindef}
\end{equation}
with $P$ the left stochastic matrix provided by \mbox{$\ma\coarser\mb$}. Therefore, $X\to J\to I$ because of \mbox{$p_{xij} = p_{j|i}p_{i|x}p_x$}. 
Moreover, 
\begin{align}
    &p_{xi} = \sum_j P_{ji}\pb_{xi} = \pb_{xi},\nonumber\\
    &p_{xj} = \sum_i P_{ji} \braket{x|\pob_i|x} p_x = \braket{x|\poa_j|x} p_x  = \pa_{xj},
\end{align}
where the second line uses the relation between POVM elements due to $\ma\coarser\mb$. Finally, theorem \ref{thm:DPIMI} implies $I(\pa_{x j})\leq I(\pb_{x i})$, achieving the equality if and only if
\begin{equation}
    p_{xij} = p_{i|j}p_{xj},
\end{equation}
which implies
\begin{align}
    p_{xi} &= \sum_j p_{i|j}p_{xj}.
\end{align}
Now, choosing $\hat\rho = \ketbra{x}{x}$ and using the definition of $p_{xi}$ and $p_{xj}$ given in Eq.~\eqref{eq:jointdef},
\begin{align}
    \braket{x|\pob_i|x}&=\braket{x|\sum_j p_{i|j}\poa_j|x}. \label{eq:expectation}
\end{align}
Since $\ket{x}$ can be any pure state, then
\begin{equation}
    \pob_i = \sum_j p_{i|j}\poa_j,
\end{equation}
which means $\ma\finer\mb$, with the conditional probabilities $p_{i|j}$ forming a left stochastic matrix. Therefore, achieving the equality in $I(\pa_{x j})\leq I(\pb_{x i})$ implies that $\ma\finer\mb$. The converse implication is
\begin{equation}
    \ma\finer\mb,\ \ma\coarser\mb \implies I(\pa_{x j})=I(\pb_{x i}),\  \forall\hat\rho,
\end{equation}
which is also true because of the already proven Eq.~\eqref{eq:lessinfo}.
\end{proof}

\section{Proving the Properties of Coarser in Subspaces}
\label{apx:subspaces}
This appendix contains the proofs that are not included in section \ref{sec:CoarserAndProcessing} regarding the properties of coarser in subspaces.

\subsection{Proof of Lemma \ref{lemma:processingmakescoarsersub}}
\label{apx:processingmakescoarsersub}
$(\ref{item:coarsersub})\implies (\ref{item:probrelsub})$. The condition $\ma\pcoarser{\mathcal{G}}\mb$ implies that $\forall j\in \oa(\mathcal{G})$,
\begin{align}
      \pa_j &=\sum_{i\in \ob(\mathcal{G})} \widetilde P_{ji}\pb_i ,\ \forall \hat\rho\in \mathcal{D}(\mathcal{G}) \label{eq:coarsersubA},\\
      \Va_j &\geq \sum_{i\in \ob(\mathcal{G})} \widetilde P_{ji} \Vb_i,\label{eq:coarsersubB}
\end{align}
where the first line came from using that $\hat P_{\mathcal{G}}\hat\rho\hat P_{\mathcal{G}}=\hat\rho$ for  $\hat\rho\in \mathcal{D}(\mathcal{G})$. Moreover, $\widetilde P$ is a left stochastic $\oa(\mathcal{G})\times\ob(\mathcal{G})$ matrix. The next step is to extend $\widetilde P$ into a $P$ which is defined even if $i\not\in\ob(\mathcal{G})$ or $j\not\in\oa(\mathcal{G})$,
\begin{align}
    P_{ji} &=
    \begin{cases}
    \widetilde P_{ji} & i\in \ob(\mathcal{G}), \  j\in \oa(\mathcal{G}), \\
    0 & i\in \ob(\mathcal{G}), \  j\not\in \oa(\mathcal{G}),\\
    C_j &  i\not\in \ob(\mathcal{G}),
    \end{cases} \nonumber\\
    C_j &= \frac{\Va_j - \sum_{i\in \ob(\mathcal{G})} P_{ji} \Vb_i}{\sum_{i\not\in\ob(\mathcal{G})} \Vb_i}.
\end{align}
Let us check that $P$ fulfills the requirements in $(\ref{item:probrelsub})$, starting with $P$ being left stochastic. $P_{ji}\geq0$ because $\widetilde P_{ji}\geq0$ and $C_j\geq0$. In turn, $C_j\geq0$ because of
\begin{itemize}
    \item Eq.~\eqref{eq:coarsersubB}, when $j\in\oa(\mathcal{G})$.
    \item $C_j=\Va_j/\sum_{i\not\in\ob(\mathcal{G})} \Vb_i\geq 0$, when $j\not\in\oa(\mathcal{G})$.
\end{itemize} 
Moreover, the columns of $P$ sum to one,
\begin{equation}
    \sum_j P_{ji} = \sum_j \widetilde P_{ji} = 1,\ \forall i\in\ob(\mathcal{G}),
\end{equation}
\begin{equation}
    \sum_j P_{ji} = \frac{\sum_j\Va_j - \sum_{i\in \ob(\mathcal{G})}\Vb_i}{\sum_{i\not\in\ob(\mathcal{G})} \Vb_i}=1 , \ \forall i\not\in\ob(\mathcal{G}),
\end{equation}
due to $\sum_j \Va_j = \sum_i \Vb_i$, i.e. $V_{\text{tot}}$ remaining constant. Therefore, $P$ is left stochastic. The next requirement is to let $\hat\rho\in\mathcal{D}(\mathcal{G})$ and then show $\pa_j = \sum_i P_{ji}\pb_i$. Using that $\pb_i = 0$ for $i\not\in\ob(\mathcal{G})$,
\begin{equation}
        \sum_i P_{ji}\pb_i = \begin{cases}
        \sum_{i\in \ob(\mathcal{G})} \widetilde P_{ji}\pb_i & j\in \oa(\mathcal{G}), \\
        0 & j\not\in \oa(\mathcal{G}),
    \end{cases}\label{eq:cond1}
\end{equation}
which is equal to $\pa_j$ because of Eq.~\eqref{eq:coarsersubA} when  $j\in\oa(\mathcal{G})$ and $\pa_j = 0$ when $j\not\in\oa(\mathcal{G})$. The requirement for the volumes is also fulfilled because
\begin{align}
    \sum_i P_{ji}\Vb_i &= \sum_{i\in\ob(\mathcal{G})} P_{ji}\Vb_i + C_j\sum_{i\not\in\ob(\mathcal{G})}\Vb_i
    \nonumber\\
    &= \Va_j.
\end{align}
Therefore, $(\ref{item:probrelsub})$ holds, completing the first implication.

$(\ref{item:coarsersub})\impliedby (\ref{item:probrelsub})$. The condition $(\ref{item:probrelsub})$ provides a $P$ fulfilling Eqs.~\eqref{eq:probrelsub} and \eqref{eq:volrelsub}. Denote by $\widetilde P$ the restriction of $P$ to $i\in\ob(\mathcal{G})$, $j\in\oa(\mathcal{G})$. Proving $(\ref{item:coarsersub})\impliedby (\ref{item:probrelsub})$ will consist in showing that $\widetilde P$ fulfills the requirements of $\ma\pcoarser{\mathcal{G}}\mb$. The first step is to prove that $\widetilde P$ is left stochastic. Consider any
$j\not\in\oa(\mathcal{G})$ and use Eq.~\eqref{eq:probrelsub}, $\pb_i=0$ for $i\not\in\ob(\mathcal{G})$ and \mbox{$\pa_j=0$} for $j\not\in\oa(\mathcal{G})$ to obtain
\begin{equation}
    \sum_{i\in\ob(\mathcal{G})} P_{ji} \pb_i = \sum_{i} P_{ji} \pb_i = \pa_j = 0  ,\ \forall\hat\rho\in\mathcal{D}(\mathcal{G}),
\end{equation}
This equation, together with that there is always a \mbox{$\hat\rho\in\mathcal{D}(\mathcal{G})$} that makes $\pb_i\neq0$ when $i\in\ob(\mathcal{G})$ and with that $P_{ji}\geq0$, implies
\begin{equation}
    P_{ji} = 0 ,\  \forall i\in \ob(\mathcal{G}),\   \forall j \not\in \oa(\mathcal{G}).
\end{equation}
This identity in turn implies
\begin{equation}
    \sum_{j\in\oa(\mathcal{G})} \widetilde P_{ji} = 1, \  \forall i\in\ob(\mathcal{G}),
\end{equation}
which means that $\widetilde P$ is left stochastic when joined with the fact that $\widetilde P_{ji}\geq 0$ because $P_{ji}\geq 0$. The POVM elements' condition, Eq.~\eqref{eq:poCond}, is shown starting from
\begin{align}
    \pa_j &= \sum_i P_{ji}\pb_i = \sum_{i\in\ob(\mathcal{G})} P_{ji}\pb_i ,\  \forall\hat\rho\in\mathcal{D}(\mathcal{G}).
\end{align}
Now, rewriting in terms of the POVM elements,
\begin{align}
     &\Tr[\poa_j\hat\rho] = \Tr\Big[\sum_{i\in\ob(\mathcal{G})} P_{ji}\pob_i\hat\rho\Big],\ \forall \hat\rho \in \mathcal{D}(\mathcal{G})\nonumber\\
     &\implies \hat P_{\mathcal{G}}\poa_j\hat P_{\mathcal{G}} = \sum_{i\in\ob(\mathcal{G})} P_{ji}\hat P_{\mathcal{G}}\pob_i \hat P_{\mathcal{G}}.
\end{align}
The requirement over the volumes, Eq.~\eqref{eq:volCond}, is also fulfilled,
\begin{equation}
    \Va_j = \sum_i P_{ji}\Vb_i\geq \sum_{i\in\ob(\mathcal{G})} P_{ji}\Vb_i,
\end{equation}
completing all the requirements for $\ma\pcoarser{\mathcal{G}}\mb$.

Finally, the proof for $(\ref{item:probrelsub})\iff (\ref{item:ItoJsub})$ is almost the same as for lemma \ref{lemma:processingmakescoarser}. These conditions are seen to be equivalent by choosing the same transition matrix $P$ for both. The only difference with the previous proof is using that 
\begin{equation}
    \Va_j = \pa_jV_{\text{tot}}, \ \Vb_i = \pb_iV_{\text{tot}}, \text{ when } \hat\rho = \hat\rho^{\text{id}},
\end{equation}
to obtain Eq.~\eqref{eq:volrelsub} from $(\ref{item:ItoJsub})$ with $\hat\rho^{\text{id}}$ and vice-versa.

\subsection{Proof of Theorem \ref{thm:CoarserObservationalSub}}
\label{apx:CoarserObservationalSub}
\begin{proof}
The proof is almost the same as for theorem \ref{thm:CoarserObservational} after substituting the lemma \ref{lemma:processingmakescoarser} with the corresponding lemma \ref{lemma:processingmakescoarsersub}. The only difference arises when proving that achieving $S_{\ma}=S_{\mb}$ implies $\ma\finer\mb$. The following extra step has to be shown,
\begin{equation}
    \sum_j \tilde P_{ij} \Va_j = \sum_j P_{ji}\frac{\Vb_i}{\Va_j} \Va_j= \Vb_i,
\end{equation}
because this equation is required by item $(\ref{item:probrelsub})$ of lemma \ref{lemma:processingmakescoarsersub}. Besides this addition, the proof proceeds as in theorem \ref{thm:CoarserObservational}.
\end{proof}

\subsection{Proof of Theorem \ref{thm:CoarserProcessingSub}}
\label{apx:CoarserProcessingSub}
\begin{proof}
The proof proceeds analogously to the proof of theorem \ref{thm:CoarserProcessing}. The major difference is the use of lemma \ref{lemma:processingmakescoarsersub} together with $\ma\pcoarser{\mathcal{G}}\mb$ to obtain an adequate $P$. This $P$ is defined for any possible value of $i$ and $j$, allowing to build the Markov chain $X\to I \to J$ with the probabilities defined in Eq.~\eqref{eq:chaindef}. The rest of the proof proceeds without changes except when dealing with the case $I(\pa_{x j})=I(\pb_{x i})$. Since now $\hat\rho \in \mathcal{D}(\mathcal{G})$, the $\hat\rho=\ketbra{x}{x}$ can only be taken with $\ket{x}\in\mathcal{G}$. Therefore, the Eq.~\eqref{eq:expectation} only implies that 
\begin{equation}
    \hat P_{\mathcal{G}}\poa_i\hat P_{\mathcal{G}} =\sum_{j\in \oa(\mathcal{G})} p_{i|j}\hat P_{\mathcal{G}}\pob_j \hat P_{\mathcal{G}} ,\  \forall i\in \ob(\mathcal{G}), \label{eq:reversed}
\end{equation}
as required. This equation not enough to state that $\ma\pfiner{\mathcal{G}}\mb$, because the condition over the volumes is missing. However, it turns out that this condition on the volumes is not needed to prove the inequality between mutual informations. Therefore, Eq.~$\eqref{eq:reversed}$ implies $I(\pa_{x j})\geq I(\pb_{x i})$, which allows to complete the proof.
\end{proof}

\subsection{Proving the Observations in Theorem \ref{thm:preserved}}
\label{apx:observations}
To prove these observations it will be useful to first prove the two lemmas provided below. 
\begin{lemma} The POVM elements of $\m$ outside of $\om(\mathcal{G})$ are null in the subspace $\mathcal{G}$, and vice-versa.
\begin{equation}
    \po_i \hat P_{\mathcal{G}} = 0 \iff i \not\in \om(\mathcal{G}).
\end{equation}
\label{lemma:indexs}
\end{lemma}
\begin{proof}
Let $\m$ be a measurement and $\mathcal{G}$ a subspace of the Hilbert space. First, 
\begin{equation}
    \po_i \hat P_{\mathcal{G}} = 0 \implies i \not\in \om(\mathcal{G}).
\end{equation}
because 
\begin{align}
    &\po_i\ket{\psi} = \po_i\hat P_{\mathcal{G}}\ket{\psi} = 0 ,\ \forall \ket{\psi}\in \mathcal{G} \nonumber \\ 
    &\implies p_i = 0,\  \forall \ket{\psi}\in \mathcal{G},
\end{align}
which means $i \not\in \om(\mathcal{G})$.

Second, let us prove
\begin{equation}
    \po_i \hat P_{\mathcal{G}} = 0 \impliedby i \not\in \om(\mathcal{G}).
\end{equation}
Let \mbox{$\hat P_{\mathcal{G}} = \sum_x \ketbra{x}{x}$} be the decomposition in eigenvectors. Then, 
\begin{equation}
    \Tr[\po\hat P_{\mathcal{G}}] = \sum_x \braket{x|\po_i|x} = 0,
\end{equation}
where we used that $\ket{x}\in\mathcal{G}$ and since $i \not\in \om(\mathcal{G})$, then \mbox{$\braket{x|\po_i|x} = 0$}. Finally, since $\Tr[\po\hat P_{\mathcal{G}}] = 0$, the previous lemma \ref{lemma:projandpo} implies $\po\hat P_{\mathcal{G}} = 0$, completing the proof.
\end{proof}

\begin{lemma} Some entries in the left stochastic matrix might be zero due to $\oa$ and $\ob$. Let $\mathcal{F}$ and $\mathcal{G}$ be subspaces of $\mathcal{H}$. Consider 
\begin{align}
    \tilde J &= \oa(\mathcal{G}) \setminus \oa(\mathcal{F}), \\
    I^\cap &= \ob(\mathcal{G}) \cap \ob(\mathcal{F}).
\end{align}
with $X\setminus Y$ the elements of $X$ that do not belong to $Y$. For any left stochastic matrix $P$ and $j\in \tilde J$,
\begin{equation}
   \Tr[\hat R_j\hat P_{\mathcal{F}}]\geq 0  \implies P_{ji} = 0\   i \in I^\cap, \ \hat R_j\hat P_{\mathcal{F}} = \hat P_{\mathcal{F}}\hat R_j = 0,
\end{equation}
where
\begin{equation}
    \hat R_j = \poa_j - \sum_{i\in \ob(\mathcal{G})} P_{ji} \pob_i.
\end{equation}
\label{lemma:ZeroPij}
\end{lemma}
\begin{proof}
Let $j\in \tilde J$. Then, $j \not\in \oa(\mathcal{F})$ and applying the lemma \ref{lemma:indexs},
\begin{equation}
    \poa_j \hat P_{\mathcal{F}} = 0.
\end{equation}
Therefore, from the definition of $\hat R_j$,
\begin{align}
    \hat R_j \hat P_{\mathcal{F}} = - \sum_{i\in \ob(\mathcal{G})} P_{ji}  \pob_i \hat P_{\mathcal{F}}\nonumber\\
    = - \sum_{i\in I^\cap} P_{ji}  \pob_i \hat P_{\mathcal{F}}
    \label{eq:projectF}
\end{align}
where the lemma \ref{lemma:indexs} was applied again to show that $\pob_i\hat P_{\mathcal{F}} = 0$, $\forall i\not\in \ob(\mathcal{F})$. Taking the trace, 
\begin{align}
    0&\leq -\sum_{i\in I^\cap} P_{ji} \Tr[\pob_i \hat P_{\mathcal{F}}],  \label{eq:desigual}
\end{align}
by using the hypothesis $\Tr[\hat R_j\hat P_{\mathcal{F}}]\geq 0$. Finally, since $i\in I^\cap\subseteq \ob(\mathcal{F})$, the lemma \ref{lemma:indexs} implies $\pob_{i} \hat P_{\mathcal{F}}\neq 0$. From here, lemma \ref{lemma:projandpo} implies $\Tr[\pob_{i} \hat P_{\mathcal{F}}] > 0$. Therefore, since all $P_{ji}\geq0$, the only option to fulfill Eq.~\eqref{eq:desigual} for $j\in\tilde J$ is 
\begin{equation}
    P_{ji} = 0, \   \forall i \in I^\cap,
\end{equation}
which was one of the objectives of the proof. This result, in combination with Eq.~\eqref{eq:projectF} implies $\hat R_j\hat P_{\mathcal{F}} = 0$. The proof for $\hat P_{\mathcal{F}}\hat R_j = 0$ is the same, but starting with $\hat P_{\mathcal{F}}\poa_j = 0$.
\end{proof}
Let us proceed to prove the observations in theorem \ref{thm:preserved}.
\begin{proof}
Let $\mathcal{F}\subseteq\mathcal{G}$ be two subspaces of the Hilbert space. 

First, let us prove that $\om(\mathcal{F})\subseteq \om(\mathcal{G})$ for any $\m$,
\begin{align}
    j\in \om(\mathcal{F}) &\iff \exists \ket{\psi} \in \mathcal{F}\subseteq\mathcal{G} \text{ s.t. } p_j \neq 0\nonumber \\ 
    &\implies j\in \om(\mathcal{G}),
\end{align}
where the definition \ref{def:setoutcomes} was used. 

Second, assume $\ma\overset{\mathcal{G}}{\coarser}\mb$, with a corresponding left stochastic matrix $P^{\mathcal{G}}$. Then, as already discussed in the body of the article, $\ma\overset{\mathcal{F}}{\coarser}\mb$ is true. Now, the objective is to prove that $P^{\mathcal{F}}$, defined as 
\begin{equation}
    P^{\mathcal{F}}_{ji} = P^{\mathcal{G}}_{ji} ,\  i\in \ob(\mathcal{F}),\  j\in \oa(\mathcal{F}), 
\end{equation}
is a left stochastic matrix that fulfills
\begin{align}
    \hat P_\mathcal{F} \poa_i \hat P_\mathcal{F}&= \sum_{i\in \ob(\mathcal{F})} P^{\mathcal{F}}_{ji} \hat P_\mathcal{F}\pob_i\hat P_\mathcal{F} ,\  j\in \oa(\mathcal{F}),\label{eq:PFp}\\
    \Va_j &\geq \sum_{i\in \ob(\mathcal{F})} P_{ji} \Vb_i,\  \forall j\in \oa(\mathcal{F}).\label{eq:PFV}
\end{align}
Start with $\ma\overset{\mathcal{G}}{\coarser}\mb$, which provides 
\begin{align}
    \hat P_\mathcal{G} \poa_j \hat P_\mathcal{G}&= \sum_{i\in \ob(\mathcal{G})} P^{\mathcal{G}}_{ji} \hat P_\mathcal{G}\pob_i\hat P_\mathcal{G} ,\  j\in \oa(\mathcal{G}),\label{eq:PGp}\\
    \Va_j &\geq \sum_{i\in \ob(\mathcal{G})} P_{ji} \Vb_i ,\  \forall j\in \oa(\mathcal{G}).\label{eq:PGV}
\end{align}
Multiplying by $\hat P_\mathcal{F}$ by the left and by the right on Eq.~\eqref{eq:PGp} produces Eq.~\eqref{eq:PFp}, using that
\begin{equation}
    \mathcal{F}\subseteq\mathcal{G} \implies \hat P_\mathcal{F}\hat P_\mathcal{G}= \hat P_\mathcal{G}\hat P_\mathcal{F}=\hat P_\mathcal{F}\label{eq:projinclosos},
\end{equation}
together with $\pob_i\hat P_\mathcal{F}=0$, $\forall i\in\ob(\mathcal{G})\setminus\ob(\mathcal{F})$ by lemma \ref{lemma:indexs} and that $\oa(\mathcal{F})\subseteq \oa(\mathcal{G})$. Moreover, Eq.~\eqref{eq:PGV} converts to Eq.~\eqref{eq:PFV} by removing non-negative terms on the right hand side and considering $\oa(\mathcal{F})\subseteq \oa(\mathcal{G})$. 

The last step is to show that $P^{\mathcal{F}}$ is left stochastic, which reduces to seeing
\begin{equation}
    \sum_{j\in \oa(\mathcal{F})} P^{\mathcal{F}}_{ji} = 1, \  i\in  \ob(\mathcal{F}).
\end{equation}
This identity can be derived applying lemma \ref{lemma:ZeroPij} to $P^{\mathcal{G}}$. The hypothesis of the lemma is fulfilled because of Eq.~\eqref{eq:PGp} and Eq.~\eqref{eq:projinclosos}. The consequence is that
\begin{equation}
    P_{ji}^{\mathcal{G}} = 0, \ j\in\tilde{J},\ i \in I^\cap, 
\end{equation}
where $\tilde J = \oa(\mathcal{G})\setminus \oa(\mathcal{F})$ and $I^\cap = \ob(\mathcal{F})$. Therefore,
\begin{equation}
    \sum_{j\in \oa(\mathcal{F})} P^{\mathcal{F}}_{ji} = \sum_{j\in \oa(\mathcal{G})} P^{\mathcal{G}}_{ji} = 1, \ i\in  \ob(\mathcal{F}),
\end{equation}
completing the proof that $P^{\mathcal{F}}$ defined as the restriction of $P^{\mathcal{G}}$ is a valid left stochastic matrix for $\ma\overset{\mathcal{F}}{\coarser}\mb$.
\end{proof}

\subsection{Proving that Coarseness is not Preserved under Sum of Subsets}
\label{apx:counterexempleSum}
To prove that 
\begin{equation}
    \ma \overset{\mathcal{G}}{\coarser}\mb,\ \ma \overset{\mathcal{F}}{\coarser}\mb \centernot\implies \ma \pcoarser{\mathcal{G}+\mathcal{F}}\mb,\label{eq:noSimpleSum}
\end{equation}
it is enough to provide a counterexample. The counterexample uses measurements on a two dimensional system, with
\begin{align}
    &\poa_1=\ketbra{+}{+},\quad \poa_2=\ketbra{-}{-}, \nonumber\\
    &\pob_1=\ketbra{0}{0},\quad \pob_2=\ketbra{1}{1},
\end{align}
with $\ket{\pm} = (\ket{0} \pm \ket{1})/\sqrt{2}$. Let
\begin{equation}
    \mathcal{G}=\Span(\ket{0}),\quad \mathcal{F}=\Span(\ket{1}).
\end{equation}
Then, $\ma \overset{\mathcal{G}}{\coarser}\mb$ because 
\begin{align}
    \hat P_{\mathcal{G}}\poa_1\hat P_{\mathcal{G}} &= \frac{1}{2}\ketbra{0}{0} = \frac{1}{2}\hat P_{\mathcal{G}}\pob_1\hat P_{\mathcal{G}},
\end{align}
and $1=\Va_1\geq\frac{1}{2}\Vb_1=\frac{1}{2}$. Analogously, $\ma \overset{\mathcal{F}}{\coarser}\mb$ because 
\begin{align}
    \hat P_{\mathcal{G}}\poa_2\hat P_{\mathcal{G}} &= \frac{1}{2}\ketbra{0}{0} = \frac{1}{2}\hat P_{\mathcal{G}}\pob_1\hat P_{\mathcal{G}},
\end{align}
and $1=\Va_2\geq\frac{1}{2}\Vb_1=\frac{1}{2}$. However, \mbox{$\mathcal{F}+\mathcal{G} = \mathcal{H}$}. Then, $\ma \pcoarser{\mathcal{F}+\mathcal{G}}\mb$ would require the projectors of $\ma$ to be written as a sum of the projectors of $\mb$, which is impossible. Therefore, $\ma \pcoarser{\mathcal{F}+\mathcal{G}}\mb$ is false, even though $\ma \overset{\mathcal{G}}{\coarser}\mb$ and $\ma \overset{\mathcal{F}}{\coarser}\mb$ are true, providing the counterexample that proves Eq.~\eqref{eq:noSimpleSum}.

\subsection{Proving that $P^{\mathcal{F}}$ Cannot Always be Extended}
\label{apx:counterexempleExtension}
This section provides a counterexample to show that even if $\mathcal{F}\subseteq\mathcal{G},\,\ma\pcoarser{\mathcal{G}}\mb,$ the left stochastic matrix $P^{\mathcal{F}}$ of $\ma\pcoarser{\mathcal{F}}\mb$ cannot be extended to a $P^{\mathcal{G}}$ in general.

Consider a two dimensional system with Hilbert space $\mathcal{G}=\Span(\ket{0},\ket{1})$, and the subset $\mathcal{F} = \Span(\ket{+})$, with $\ket{+} = (\ket{0} + \ket{1})/\sqrt{2}$. The measurement
$\m$ with 
\begin{equation}
    \po_1=\ketbra{0}{0},\quad \po_2=\ketbra{1}{1}, 
\end{equation}
fulfills $\m\pcoarser{\mathcal{F}}\m$ with
\begin{equation}
    P^{\mathcal{F}} = 
    \begin{bmatrix}
        0 & 1\\
        1 & 0 
    \end{bmatrix}.
\end{equation}
This $P^{\mathcal{F}}$ is adequate because
\begin{equation}
    \ketbra{+}{+}\po_1\ketbra{+}{+} = \ketbra{+}{+}\po_2\ketbra{+}{+} = \frac{1}{2}\ketbra{+}{+},
\end{equation}
which makes the condition about the POVM elements, Eq.~\eqref{eq:poCond}, true. In addition, the condition on the volumes Eq.~\eqref{eq:volCond} is readily verified using $V_1=V_2=1$. Therefore, $P^{\mathcal{F}}$ is adequate for $\m\pcoarser{\mathcal{F}}\m$. However, $P^{\mathcal{F}}$ cannot be extended to any $P^{\mathcal{G}}$ to show $\m\pcoarser{\mathcal{G}}\m$. The only possible choice would be $P^{\mathcal{G}} = P^{\mathcal{F}}$, which does not work because 
\begin{equation}
    \po_1 = \po_2,\ \po_2=\po_1,
\end{equation}
which is false. Simultaneously, $\m\pcoarser{\mathcal{G}}\m$ is true, with 
\begin{equation}
    P^{\mathcal{G}} = 
    \begin{bmatrix}
        1 & 0\\
        0 & 1 
    \end{bmatrix},
\end{equation}
completing the counterexample.

\bibliography{bibliography}

\end{document}